\documentclass[
preprint,
amsmath,amssymb,
aps,
]{revtex4-2}

\usepackage{graphicx}
\usepackage{dcolumn}
\usepackage{bm}
\usepackage{hyperref}

\usepackage{amsthm}

\usepackage{dcolumn}
\usepackage{bm}
\usepackage{hyperref}
\usepackage{subcaption}
\usepackage{caption}
\usepackage{diagbox}
\usepackage{booktabs}
\usepackage{adjustbox}
\usepackage{multirow}
\usepackage[mathlines]{lineno}
\usepackage{ragged2e}
\usepackage{pifont}
\usepackage{makecell}
\usepackage{algpseudocode}
\DeclareCaptionType{algorithm}[Algorithm][List of Algorithms]
\usepackage{soul}
\usepackage{cleveref}
\usepackage{tikz}
\usepackage{array}
\usepackage{centernot}
\usepackage{cancel}
\usepackage{enumitem}

\newtheorem{theorem}{Theorem}
\newtheorem{lemma}{Lemma}
\newtheorem{proposition}{Proposition}
\newtheorem{corollary}{Corollary}

\newtheorem{definition}{Definition}

\newcommand{\ket}[1]{\lvert #1\rangle}
\newcommand{\bra}[1]{\langle #1\rvert}
\newcommand{\braket}[2]{\langle #1\vert #2\rangle}
\newcommand{\proj}[1]{\lvert #1\rangle\!\langle #1\rvert}
\newcommand{\Tr}{\operatorname{Tr}}
\newcommand{\id}{\mathsf{I}}
\newcommand{\eps}{\varepsilon}
\newcommand{\GHZ}{\mathrm{GHZ}}
\newcommand{\BM}{\mathrm{BM}}

\newcommand{\wt}{\operatorname{wt}}

\definecolor{codegreen}{rgb}{0,0.6,0}
\definecolor{codegray}{rgb}{0.5,0.5,0.5}
\definecolor{codepurple}{rgb}{0.58,0,0.82}
\definecolor{backcolour}{rgb}{0.95,0.95,0.92}

\begin{document}


\title{Disjoint Bell measurements enable near-projective\\GHZ certification}


\author{Hyunho Cha}
\email{Contact author: ovalavo@snu.ac.kr}
\author{Jungwoo Lee}
\email{Contact author: junglee@snu.ac.kr}
\affiliation{NextQuantum Innovation Research Center, Seoul National University, Seoul 08826, Republic of Korea}


\begin{abstract}
Certifying multipartite entangled states is a basic task in quantum information processing, but the achievable copy complexity depends crucially on the measurements available to the verifier. The strongest possible certification measurement for a known pure target state \(|\psi\rangle\) is the two-outcome projector \(\{|\psi\rangle\langle\psi|,\mathsf{I}-|\psi\rangle\langle\psi|\}\), which is copy-optimal but often experimentally unrealistic or outside the intended measurement model. In this work, we introduce BM-Cert, a single-copy verification protocol for the \(n\)-qubit Greenberger--Horne--Zeilinger (GHZ) state using only disjoint two-qubit Bell-basis measurements, together with one single-qubit \(X\)-basis measurement when \(n\) is odd. Surprisingly, a simple combinatorial effect yields perfect completeness and a verification spectral gap \(\nu_\mathrm{BM}(n)=1-O(1/n)\), so our depth-2 protocol already approaches the ideal projective verification asymptotically as \(n\) grows. This contrasts with local Pauli GHZ verification, whose optimal spectral gap remains bounded away from \(1\). Thus, allowing only two-qubit entangling measurements on disjoint pairs is enough to achieve asymptotically ideal projective certification. The same Bell-matching outcomes also yield BM-Fid, an unbiased estimator of the GHZ fidelity whose leading Hoeffding coefficient in the sample complexity tends to the ideal value achieved by direct projection. For the open-boundary linear nearest neighbor setting, we further introduce Brick-Cert, a disjoint 2-local certification protocol whose spectral gap \(4/5\) is optimal within that restricted architecture.
\end{abstract}

\maketitle


\section{Introduction}

A quantum computer or quantum network is often designed to output a particular state. In an experiment, noise and calibration errors can make the actually produced state differ from the intended one \cite{geller2020rigorous, cai2023quantum, dasgupta2024impact, khan2024multiaxis, yu2025efficient, haupt2026mitigating}. \emph{State-preparation certification} asks for a statistically justified assessment of whether the produced states are close to the target state \cite{somma2006lower, aolita2015reliable, buadescu2019quantum, eisert2020quantum, kliesch2021theory, huang2025certifying}.
The target in this work is the Greenberger--Horne--Zeilinger (GHZ) state
\(\ket{G_n}=(\ket{0^n}+\ket{1^n})/\sqrt2\), a standard multipartite entangled state with roles in quantum networks, secret sharing, teleportation, nonlocality, and entanglement certification \cite{greenberger1989going, greenberger1990bell, mermin1990extreme, karlsson1998quantum, bouwmeester1999observation, hillery1999quantum, dur2000three, pan2000experimental, toth2005entanglement, hahn2020anonymous, li2020optimal, han2021optimal, pickston2023conference}. Full tomography learns an entire density matrix and is generally far more expensive than merely checking closeness to one target \cite{james2001measurement, blume2010optimal, cramer2010efficient, gross2010quantum, da2011practical, christandl2012reliable, haah2016sample, o2016efficient}. Direct fidelity estimation \cite{da2011practical, flammia2011direct, zhang2021direct, barbera2025sampling, cha2025efficient, cha2026operator}, shadow tomography \cite{aaronson2018shadow, huang2020predicting, zhang2021experimental, nguyen2022optimizing, hu2023classical, ippoliti2024classical, king2025triply}, and quantum state verification \cite{pallister2018optimal,takeuchi2018verification,liu2019efficient,zhu2019efficienthyp,zhu2019efficient,li2020optimal,zhang2020experimental,govcanin2022sample,yu2022statistical,li2023robust,li2026universal} are major alternatives.
Other efficient schemes for characterizing structured multipartite states include entanglement detection protocols \cite{kiesel2005experimental, toth2005detecting, toth2005entanglement, chen2007experimental, vallone2007realization, gong2019genuine, thomas2022efficient} and permutationally invariant tomography \cite{toth2010permutationally, schwemmer2014experimental}.

If there were no measurement restriction, the verifier would simply measure
\(\{\proj{G_n},\id-\proj{G_n}\}\). One copy then passes with probability exactly equal to the fidelity \(\bra{G_n}\rho\ket{G_n}\). This is information-theoretically optimal but hides almost all experimental difficulty inside one global projector. Recent work also shows that nondemolition or memory-assisted models can approach the performance of global measurements by using stronger resources \cite{liu2021universally,chen2025quantum}. We study a more restricted and explicit family:
\begin{quote}
\emph{On each copy, the verifier may only perform disjoint two-qubit Bell measurements, plus one single-qubit \(X\)-basis measurement if one qubit is unmatched. The verifier may randomize the pairing and may postprocess the classical outcomes arbitrarily, but may not measure across different copies and may not implement a global projector or inverse GHZ-preparation circuit.}
\end{quote}

\begin{figure}
    \centering
    \includegraphics[width=0.9\linewidth]{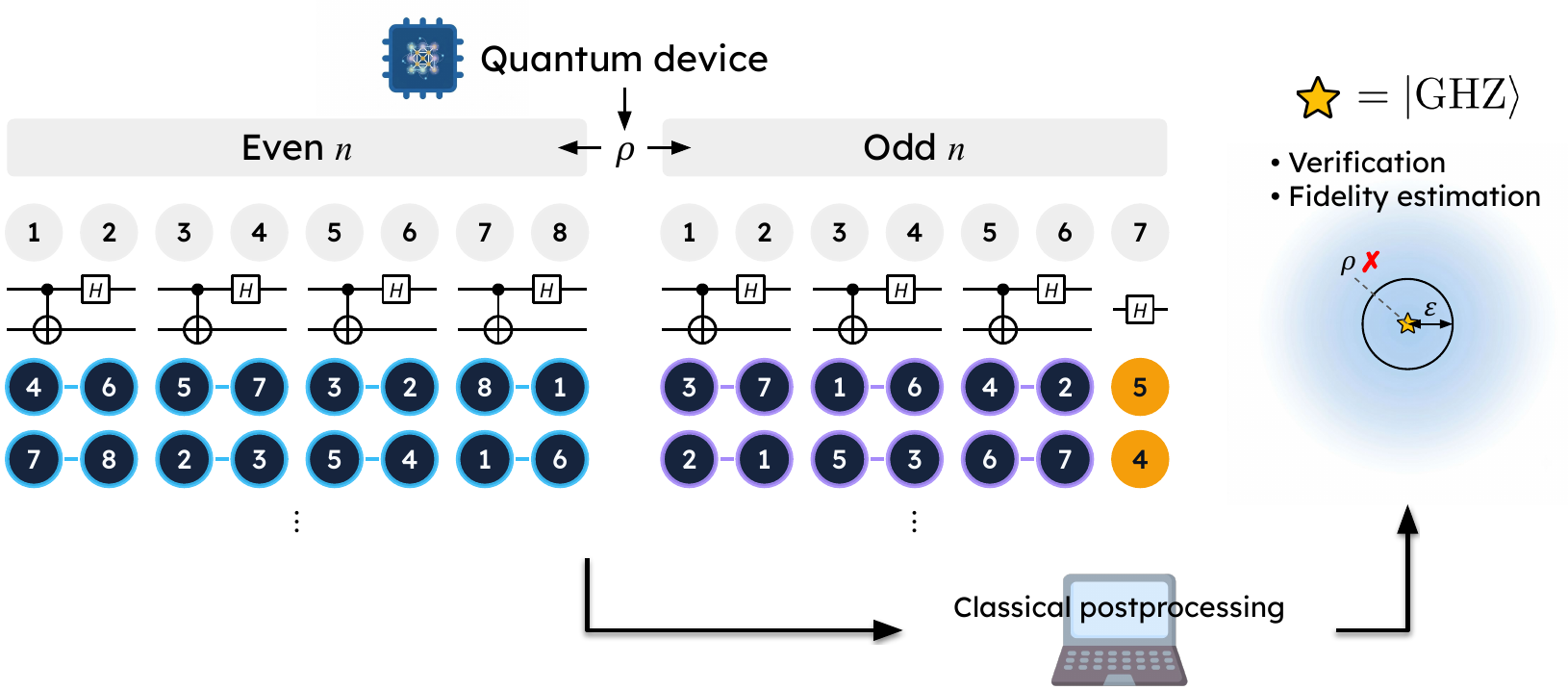}
    \caption{\justifying Schematic illustration of our measurements for GHZ certification. In each round, the protocol samples a uniformly random quasi-perfect matching of the \(n\) qubits. For every matched pair, a Bell-basis measurement is performed. For odd \(n\), one singleton qubit remains unmatched and is measured in the \(X\) basis.}
    \label{fig:schematic}
\end{figure}

Figure~\ref{fig:schematic} illustrates this restricted measurement model.
Within this family we give an optimal strategy. The main consequence is that by allowing only disjoint two-qubit Bell measurements, the Bell-matching strategy achieves spectral gap \(\nu_{\BM}(n)=1-O(1/n)\), which converges to the ideal projective value \(1\) as \(n\) grows. This gives a qualitative improvement over the local Pauli benchmark for qubit GHZ states, whose optimal spectral gap is \(2/3\) \cite{li2020optimal,dangniam2020optimal}. In particular, the Bell-matching strategy has a strictly larger spectral gap for all odd \(n\ge5\) and even \(n\ge6\).

Beyond this certification result, we also ask whether the measurement outcome contains enough information to estimate the target fidelity, rather than merely to accept or reject a copy.
We answer this by adding a postprocessing rule to the random matching outcomes, producing BM-Fid, an unbiased estimator whose Hoeffding coefficient in the sample complexity approaches the ideal value.
We then consider the complementary question of limited connectivity and show, through Brick-Cert, that on the open path \(1-2-\cdots-n\) one can achieve the optimal disjoint 2-local spectral gap \(4/5\).

\section{Preliminaries}

For a bit string \(a=(a_1,\ldots,a_n)\in\{0,1\}^n\), write
\[
\ket{a}=\ket{a_1}\otimes\cdots\otimes\ket{a_n}.
\]

\subsection{Fidelity}

\begin{definition}[Fidelity with a pure target]
Let \(\ket\psi\) be a unit vector and let \(\rho\) be a density matrix on the same Hilbert space. The fidelity of \(\rho\) with \(\ket\psi\) is
\[
F_\psi(\rho):=\bra\psi\rho\ket\psi=\Tr(\proj\psi\rho).
\]
The infidelity is \(1-F_\psi(\rho)\).
\end{definition}

The fidelity lies in \([0,1]\). It equals \(1\) exactly when \(\rho=\proj\psi\), and it equals \(0\) when \(\rho\) is completely supported on the subspace orthogonal to \(\ket\psi\).

\subsection{Pauli operators}

The Pauli \(X\) and \(Z\) matrices are
$
X=(\begin{smallmatrix}0&1\\1&0\end{smallmatrix})
$
and
$
Z=(\begin{smallmatrix}1&0\\0&-1\end{smallmatrix}).
$
If an operator acts on qubit \(i\), we write it with a subscript. For example,
\(Z_iZ_j\) means \(Z\) on qubits \(i\) and \(j\), and the identity on all other qubits. The operator
$
X^{\otimes n}:=X_1X_2\cdots X_n
$
flips every computational-basis bit.

\subsection{The GHZ target state}

\begin{definition}
For \(n\ge2\), the \(n\)-qubit GHZ state is
\[
\ket{G_n}:=\frac{\ket{0^n}+\ket{1^n}}{\sqrt2},
\]
where \(0^n=00\cdots0\) and \(1^n=11\cdots1\).
\end{definition}

The GHZ state is stabilized by the following observables:
\[
X^{\otimes n}\ket{G_n}=\ket{G_n},
\qquad
Z_iZ_j\ket{G_n}=\ket{G_n}\quad\text{for all }i\ne j.
\]

\subsection{Measurements, pass effects, and verification operators}

A two-outcome quantum test is described by a positive semidefinite operator \(E\) satisfying \(0\le E\le\id\). The probability that state \(\rho\) passes the test is \(\Tr(E\rho)\). The complementary failure probability is \(1-\Tr(E\rho)=\Tr((\id-E)\rho)\).

\begin{definition}
Let \(\ket\psi\) be the target. A test effect \(E\) has perfect completeness for \(\ket\psi\) if
\[
E\ket\psi=\ket\psi.
\]
Equivalently, \(\ket\psi\) passes the test with probability \(1\).
\end{definition}

A randomized verification strategy chooses a test effect \(E_s\) with probability \(p_s\). Its average pass effect, or verification operator, is
\[
\Omega:=\sum_s p_sE_s.
\]
If every \(E_s\) has perfect completeness, then \(\Omega\ket\psi=\ket\psi\).

\begin{definition}
Let \(\Omega\) be a verification operator with \(\Omega\ket\psi=\ket\psi\). Its second eigenvalue relative to \(\ket\psi\) is
\[
\beta(\Omega):=\bigl\| (\id-\proj\psi)\Omega(\id-\proj\psi)\bigr\|,
\]
where \(\|\cdot\|\) is the largest singular value, which equals the largest eigenvalue for positive semidefinite matrices. The spectral gap is
\[
\nu(\Omega):=1-\beta(\Omega).
\]
\end{definition}

The spectral gap is the standard copy-complexity parameter in quantum state verification \cite{pallister2018optimal,zhu2019efficient,yu2022statistical}, as explained by the following lemma.

\begin{lemma}[\cite{zhu2019efficient}]
\label{lem:many-copy}
Suppose the verifier tests \(N\) independent copies, using the same verification operator \(\Omega\) with spectral gap \(\nu\), and accepts only if every copy passes. Let the \(j\)-th copy have state \(\rho_j\) and infidelity
\[
\eps_j:=1-F_\psi(\rho_j).
\]
If the average infidelity obeys
\[
\bar\eps:=\frac1N\sum_{j=1}^N\eps_j\ge\eps,
\]
then the probability that all \(N\) tests pass is at most
\[
(1-\nu\eps)^N.
\]
Consequently, a significance level \(\delta\in(0,1)\) is guaranteed whenever
\[
N\ge
\left\lceil
\frac{\ln(1/\delta)}{-\ln(1-\nu\eps)}
\right\rceil .
\]
\end{lemma}

\section{Main algorithm}

\subsection{Bell measurements}

The four two-qubit Bell states are
\[
\ket{\Phi^+}=\frac{\ket{00}+\ket{11}}{\sqrt2},\quad
\ket{\Phi^-}=\frac{\ket{00}-\ket{11}}{\sqrt2},
\]
\[
\ket{\Psi^+}=\frac{\ket{01}+\ket{10}}{\sqrt2},\quad
\ket{\Psi^-}=\frac{\ket{01}-\ket{10}}{\sqrt2}.
\]
A Bell-basis measurement on qubits \(i,j\) distinguishes these four states \cite{bennett1993teleporting, nielsen2010quantum}. Equivalently, it jointly measures the commuting observables \(Z_iZ_j\) and \(X_iX_j\). Its classical outcome is a pair of signs
\[
(z_{ij},x_{ij})\in\{+1,-1\}^2,
\]
where \(z_{ij}\) is the measured eigenvalue of \(Z_iZ_j\), and \(x_{ij}\) is the measured eigenvalue of \(X_iX_j\).

\begin{definition}
Let \([n]=\{1,2,\ldots,n\}\). A quasi-perfect matching \(Q\) of \([n]\) is:
\begin{itemize}[leftmargin=2em]
\item if \(n\) is even, a perfect matching, meaning a partition of \([n]\) into \(n/2\) unordered pairs;
\item if \(n\) is odd, a near-perfect matching, meaning a partition of \([n]\) into \((n-1)/2\) unordered pairs and one singleton.
\end{itemize}
We write \(M(Q)\) for the set of pairs in \(Q\). When \(n\) is odd, we write \(s(Q)\) for the unique singleton.
\end{definition}

The number of perfect matchings on \(n\) labeled vertices, for even \(n\), is
\[
(n-1)!!=(n-1)(n-3)\cdots3\cdot1.
\]
The number of near-perfect matchings on \(n\) labeled vertices, for odd \(n\), is
\[
n(n-2)!!=n(n-2)(n-4)\cdots3\cdot1.
\]
The double-factorial convention \((-1)!!=1\) will be used later.

\subsection{The BM-Cert algorithm}

\begin{algorithm}[t]
\caption{\justifying BM-Cert for \(\ket{G_n}\), with \(n\ge3\)}
\label{alg:bm-cert}
\begin{algorithmic}[1]
\Require Infidelity threshold \(\eps\in(0,1)\) and significance level \(\delta\in(0,1)\)

\State Set
\[
\nu_{\BM}(n)=
\begin{cases}
1-\dfrac{1}{n-1}, & n \text{ even},\\[5pt]
1-\dfrac{1}{n}, & n \text{ odd},
\end{cases}
\]
\Statex and
\[
N=\left\lceil\frac{\ln(1/\delta)}{-\ln(1-\nu_{\BM}(n)\eps)}\right\rceil.
\]

\For{each round \(t=1,\ldots,N\)}
    \State Sample a uniformly random quasi-perfect matching \(Q_t\) of \([n]\)

    \For{every pair \(\{i,j\}\in M(Q_t)\)}
        \State Perform a Bell-basis measurement on qubits \(i,j\)
        \State Record its outcome \(z_{ij},x_{ij}\in\{+1,-1\}\)
    \EndFor

    \If{\(n\) is odd}
        \State Measure the singleton qubit \(s(Q_t)\) in the \(X\) basis
        \State Record its outcome \(x_{s(Q_t)}\in\{+1,-1\}\)
    \EndIf

    \State Accept this round if and only if
    \Statex
    \[
     z_{ij}=+1\ \forall\{i,j\}\in M(Q_t),
    \]
    \Statex
    \[
    \begin{cases}
    \prod_{\{i,j\}\in M(Q_t)}x_{ij}=+1, & n \text{ even},\\[5pt]
    x_{s(Q_t)}\prod_{\{i,j\}\in M(Q_t)}x_{ij}=+1, & n \text{ odd}.
    \end{cases}
    \]
\EndFor

\Statex \textbf{Final decision:} Accept the preparation only if every one of the \(N\) rounds accepts
\end{algorithmic}
\end{algorithm}

Our main algorithm is presented in Algorithm~\ref{alg:bm-cert}. To obtain a uniform quasi-perfect matching, randomly permute the \(n\) vertices and then pair consecutive vertices. If \(n\) is odd, leave the last vertex as the singleton. Every quasi-perfect matching appears the same number of times under this procedure.

\section{Operator form of one Bell-matching test}

For a pair \(e=\{i,j\}\), define the projector onto the \(+1\) eigenspace of \(Z_iZ_j\) by
\[
P^Z_e:=\frac{\id+Z_iZ_j}{2}.
\]
Define also the projector onto the \(+1\) eigenspace of the global \(X\)-parity by
\[
P_X:=\frac{\id+X^{\otimes n}}{2}.
\]
For a quasi-perfect matching \(Q\), define
\[
\Pi_Q:=P_X\prod_{e\in M(Q)}P^Z_e.
\]
All projectors in this product commute, so the product is again a projector.

\begin{lemma}
\label{lem:pass-effect}
For every quasi-perfect matching \(Q\), the pass effect of the BM-Cert measurement with matching \(Q\) is
\[
\Pi_Q=P_X\prod_{e\in M(Q)}P^Z_e.
\]
\end{lemma}

\begin{proof}
We first consider even \(n\). Let the pairs of \(Q\) be \(e_1,\ldots,e_m\), where \(m=n/2\). The Bell measurement on pair \(e_r\) jointly measures \(Z_{e_r}:=Z_iZ_j\) and \(X_{e_r}:=X_iX_j\) for \(e_r=\{i,j\}\). The round accepts those outcomes for which
\[
Z_{e_r}=+1\quad\text{for all }r,
\qquad
\prod_{r=1}^m X_{e_r}=+1.
\]
The first condition contributes the projector \(\prod_{r=1}^mP^Z_{e_r}\). Because the pairs are disjoint,
\[
\prod_{r=1}^m X_{e_r}=X_1X_2\cdots X_n=X^{\otimes n},
\]
so the second condition contributes \(P_X=(\id+X^{\otimes n})/2\). These observables commute, and hence the joint pass effect is
\[
P_X\prod_{r=1}^mP^Z_{e_r}=\Pi_Q.
\]

Now consider odd \(n\). Let \(e_1,\ldots,e_m\) be the pairs, where \(m=(n-1)/2\), and let \(s=s(Q)\) be the singleton. The pair measurements again impose \(Z_{e_r}=+1\) for every \(r\). The total \(X\)-parity condition is now
\[
X_s\prod_{r=1}^mX_{e_r}=+1.
\]
Since the pairs together with the singleton cover every qubit exactly once,
\[
X_s\prod_{r=1}^mX_{e_r}=X_1X_2\cdots X_n=X^{\otimes n}.
\]
Thus the same joint projector \(P_X\prod_{e\in M(Q)}P^Z_e\) is obtained.
\end{proof}

\begin{lemma}
\label{lem:perfect-completeness}
For every quasi-perfect matching \(Q\),
\[
\Pi_Q\ket{G_n}=\ket{G_n}.
\]
Therefore BM-Cert has perfect completeness.
\end{lemma}

\begin{proof}
The GHZ state obeys
\[
X^{\otimes n}\ket{G_n}=\ket{G_n}
\]
because \(X^{\otimes n}\ket{0^n}=\ket{1^n}\) and \(X^{\otimes n}\ket{1^n}=\ket{0^n}\). Therefore
\[
P_X\ket{G_n}=\ket{G_n}.
\]
For any pair \(e=\{i,j\}\), the two strings \(0^n\) and \(1^n\) have equal bits at positions \(i,j\). Hence
\[
Z_iZ_j\ket{0^n}=\ket{0^n},
\qquad
Z_iZ_j\ket{1^n}=\ket{1^n},
\]
and therefore
\[
Z_iZ_j\ket{G_n}=\ket{G_n},
\qquad
P^Z_e\ket{G_n}=\ket{G_n}.
\]
Because every factor in \(\Pi_Q=P_X\prod_{e\in M(Q)}P^Z_e\) fixes \(\ket{G_n}\), their product fixes \(\ket{G_n}\).
\end{proof}

The BM-Cert verification operator is the uniform average
\[
\Omega_{\BM}:=\mathbb E_Q\,\Pi_Q,
\]
where \(Q\) is a uniformly random quasi-perfect matching.

\section{Diagonalization in a GHZ basis and the spectral gap}

For a bit string \(a\in\{0,1\}^n\), define its bitwise complement by
\[
\bar a:=a\oplus 1^n,
\]
where \(\oplus\) means bitwise addition modulo \(2\). Thus \(\bar a_i=1-a_i\). Choose one representative from each unordered pair \(\{a,\bar a\}\), and call the representative set \(R\). For each \(a\in R\) and each sign \(\tau\in\{+1,-1\}\), define
\begin{equation}
\label{eq:ghz-basis-def}
\ket{a,\tau}:=\frac{\ket a+\tau\ket{\bar a}}{\sqrt2}.
\end{equation}
The target is \(\ket{G_n}=\ket{0^n,+}\) when the representative of \(\{0^n,1^n\}\) is chosen to be \(0^n\).

\begin{lemma}[GHZ-basis eigenvectors]\label{lem:ghz-basis-eigen}
The vectors \(\ket{a,\tau}\), with \(a\in R\) and \(\tau\in\{+1,-1\}\), form an orthonormal basis of \((\mathbb C^2)^{\otimes n}\). Moreover,
\begin{equation}
\label{eq:apply-x-to-ghz-basis}
X^{\otimes n}\ket{a,\tau}=\tau\ket{a,\tau},
\end{equation}
and for every pair \(e=\{i,j\}\),
\begin{equation}
\label{eq:apply-zz-to-ghz-basis}
Z_iZ_j\ket{a,\tau}=(-1)^{a_i\oplus a_j}\ket{a,\tau}.
\end{equation}
\end{lemma}

\begin{proof}
First, for a fixed pair \(\{a,\bar a\}\), the two vectors
\[
\frac{\ket a+\ket{\bar a}}{\sqrt2},
\qquad
\frac{\ket a-\ket{\bar a}}{\sqrt2}
\]
are orthonormal. There are \(2^{n-1}\) complement pairs and two signs \(\tau\), hence \(2^n\) orthonormal vectors. The Hilbert space dimension is \(2^n\), so these vectors form an orthonormal basis.

Next,
\[
X^{\otimes n}\ket a=\ket{\bar a},
\qquad
X^{\otimes n}\ket{\bar a}=\ket a.
\]
Therefore
\[
X^{\otimes n}\ket{a,\tau}
=\frac{\ket{\bar a}+\tau\ket a}{\sqrt2}
=\tau\frac{\ket a+\tau\ket{\bar a}}{\sqrt2}
=\tau\ket{a,\tau},
\]
where the equality \(\tau^2=1\) was used.

For the \(Z_iZ_j\) statement, note that
\[
Z_iZ_j\ket a=(-1)^{a_i\oplus a_j}\ket a.
\]
Also,
\[
\bar a_i\oplus \bar a_j=(1-a_i)\oplus(1-a_j)=a_i\oplus a_j,
\]
so
\[
Z_iZ_j\ket{\bar a}=(-1)^{a_i\oplus a_j}\ket{\bar a}.
\]
Thus the same eigenvalue multiplies both terms in \(\ket{a,\tau}\), proving
\[
Z_iZ_j\ket{a,\tau}=(-1)^{a_i\oplus a_j}\ket{a,\tau}.
\]
\end{proof}

For a bit string \(a\), let
\[
A(a):=\{i\in[n]:a_i=1\},
\]
and let
\[
\wt(a):=|A(a)|.
\]
A pair \(\{i,j\}\) has \(a_i=a_j\) precisely when it does not cross the cut \((A(a),[n]\setminus A(a))\).

\begin{lemma}[Eigenvalue as a cut-avoidance probability]\label{lem:cut-prob}
Let \(Q\) be uniformly random among quasi-perfect matchings and let
\[
\Omega_{\BM}=\mathbb E_Q\Pi_Q.
\]
For every \(a\in R\) and \(\tau\in\{+1,-1\}\),
\[
\Omega_{\BM}\ket{a,\tau}=\lambda(a,\tau)\ket{a,\tau},
\]
where
\[
\lambda(a,\tau)=
\begin{cases}
\Pr_Q\bigl[\textnormal{every pair in }M(Q)\textnormal{ has equal }a\text{-bits}\bigr], & \tau=+1,\\[3pt]
0, & \tau=-1.
\end{cases}
\]
Equivalently, for \(\tau=+1\), \(\lambda(a,+)\) is the probability that the random matching avoids the cut \((A(a),[n]\setminus A(a))\).
\end{lemma}

\begin{proof}
By Lemma~\ref{lem:ghz-basis-eigen}, \(\ket{a,\tau}\) is an eigenvector of \(X^{\otimes n}\) with eigenvalue \(\tau\). Hence
\[
P_X\ket{a,\tau}=\frac{1+\tau}{2}\ket{a,\tau}.
\]
Thus \(P_X\ket{a,+}=\ket{a,+}\), while \(P_X\ket{a,-}=0\). This proves \(
\Pi_Q\ket{a,-}=0\) for every \(Q\), and hence \(\lambda(a,-)=0\).

Now assume \(\tau=+1\). For a pair \(e=\{i,j\}\), Lemma~\ref{lem:ghz-basis-eigen} gives
\[
Z_iZ_j\ket{a,+}=(-1)^{a_i\oplus a_j}\ket{a,+}.
\]
Therefore
\[
P^Z_e\ket{a,+}=
\begin{cases}
\ket{a,+}, & a_i=a_j,\\
0, & a_i\ne a_j.
\end{cases}
\]
Because \(\Pi_Q=P_X\prod_{e\in M(Q)}P^Z_e\), and \(P_X\ket{a,+}=\ket{a,+}\), we obtain
\[
\Pi_Q\ket{a,+}=
\begin{cases}
\ket{a,+}, & a_i=a_j\text{ for every }\{i,j\}\in M(Q),\\
0, & \text{otherwise}.
\end{cases}
\]
Averaging over uniformly random \(Q\) proves the stated eigenvalue formula. The phrase ``avoids the cut'' is the condition that no matched pair has one endpoint in \(A(a)\) and one endpoint outside \(A(a)\).
\end{proof}

\begin{lemma}[Cut-avoidance probability for even \(n\)]\label{lem:even-count}
Let \(n\ge4\) be even, and let \(Q\) be a uniformly random perfect matching of \([n]\). Fix a subset \(A\subset[n]\) with \(|A|=k\). The probability that every pair of \(Q\) is contained entirely in \(A\) or entirely in \([n]\setminus A\) is
\[
E_n(k)=
\begin{cases}
\dfrac{(k-1)!!(n-k-1)!!}{(n-1)!!}, & k\textnormal{ even},\\[10pt]
0, & k\textnormal{ odd}.
\end{cases}
\]
For nontrivial cuts, meaning \(1\le k\le n-1\), the largest possible value is
\[
\max_{1\le k\le n-1}E_n(k)=\frac{1}{n-1}.
\]
\end{lemma}

\begin{proof}
The total number of perfect matchings of \([n]\) is \((n-1)!!\). A perfect matching avoids the cut \((A,[n]\setminus A)\) when it separately matches the vertices inside \(A\) and the vertices inside \([n]\setminus A\). If \(k\) is odd, this is impossible because an odd number of vertices cannot be partitioned into pairs. Thus \(E_n(k)=0\) for odd \(k\).

If \(k\) is even, the number of perfect matchings inside \(A\) is \((k-1)!!\), and the number of perfect matchings inside the complement is \((n-k-1)!!\). Multiplying these independent choices gives
\[
E_n(k)=\frac{(k-1)!!(n-k-1)!!}{(n-1)!!}.
\]

It remains to maximize \(E_n(k)\) over \(1\le k\le n-1\). Odd \(k\) gives zero, so only even \(k\) matters. Also \(E_n(k)=E_n(n-k)\), so it is enough to consider even \(k\) with \(2\le k\le n/2\). For such \(k\), compute
\[
\frac{E_n(k+2)}{E_n(k)}
=\frac{(k+1)!!(n-k-3)!!}{(k-1)!!(n-k-1)!!}
=\frac{k+1}{n-k-1}.
\]
When \(2\le k\le n/2-1\), we have \(k+1\le n-k-1\), so this ratio is at most \(1\). Thus \(E_n(k)\) is nonincreasing as \(k\) moves from \(2\) toward the middle. By symmetry, the maximum is attained at \(k=2\) or \(k=n-2\). Finally,
\[
E_n(2)=\frac{1!!(n-3)!!}{(n-1)!!}
=\frac{1}{n-1}.
\]
Therefore \(\max_{1\le k\le n-1}E_n(k)=1/(n-1)\).
\end{proof}

\begin{lemma}[Cut-avoidance probability for odd \(n\)]\label{lem:odd-count}
Let \(n\ge3\) be odd, and let \(Q\) be a uniformly random near-perfect matching of \([n]\), meaning one singleton and \((n-1)/2\) pairs. Fix a subset \(A\subset[n]\) with \(|A|=k\). The probability that every pair of \(Q\) is contained entirely in \(A\) or entirely in \([n]\setminus A\) is
\[
O_n(k)=
\begin{cases}
\dfrac{k(k-2)!!(n-k-1)!!}{n(n-2)!!}, & k\textnormal{ odd},\\[12pt]
\dfrac{(n-k)(k-1)!!(n-k-2)!!}{n(n-2)!!}, & k\textnormal{ even}.
\end{cases}
\]
For nontrivial cuts, meaning \(1\le k\le n-1\), the largest possible value is
\[
\max_{1\le k\le n-1}O_n(k)=\frac1n.
\]
\end{lemma}

\begin{proof}
The total number of near-perfect matchings is
$
n(n-2)!!,
$
because there are \(n\) choices for the singleton and then \((n-2)!!\) perfect matchings of the remaining \(n-1\) vertices.

A near-perfect matching avoids the cut if, after removing the singleton, the remaining vertices inside \(A\) can be paired among themselves and the remaining vertices outside \(A\) can be paired among themselves.

First suppose \(k\) is odd. Then the singleton must lie in \(A\). There are \(k\) choices for the singleton, \((k-2)!!\) matchings of the remaining vertices in \(A\), and \((n-k-1)!!\) matchings of the complement. Therefore
\[
O_n(k)=\frac{k(k-2)!!(n-k-1)!!}{n(n-2)!!}
\]
when \(k\) is odd.

Now suppose \(k\) is even. Then the singleton must lie in the complement \([n]\setminus A\). There are \(n-k\) choices for the singleton, \((k-1)!!\) matchings of \(A\), and \((n-k-2)!!\) matchings of the remaining complement. Therefore
\[
O_n(k)=\frac{(n-k)(k-1)!!(n-k-2)!!}{n(n-2)!!}
\]
when \(k\) is even.

It remains to maximize \(O_n(k)\) over \(1\le k\le n-1\). We use symmetry \(O_n(k)=O_n(n-k)\).
For even \(k\), compare consecutive even values:
\[
\frac{O_n(k+2)}{O_n(k)}
=\frac{(n-k-2)(k+1)!!(n-k-4)!!}{(n-k)(k-1)!!(n-k-2)!!}
=\frac{k+1}{n-k}.
\]
As \(k\) moves from \(2\) toward the middle, this ratio is at most \(1\) until the midpoint is reached. Hence among even nontrivial \(k\), the maximum is at \(k=2\) (by symmetry, the same value is attained at \(k=n-2\)). Its value is
\[
O_n(2)=\frac{(n-2)1!!(n-4)!!}{n(n-2)!!}
=\frac{1}{n}.
\]
For odd \(k\), compare consecutive odd values:
\[
\frac{O_n(k+2)}{O_n(k)}
=\frac{(k+2)k!!(n-k-3)!!}{k(k-2)!!(n-k-1)!!}
=\frac{k+2}{n-k-1}.
\]
As \(k\) moves from \(1\) toward the middle, this ratio is at most \(1\) until the midpoint is reached. Hence among odd nontrivial \(k\), the maximum is at \(k=1\) (by symmetry, the same value is attained at \(k=n-1\)). Its value is
\[
O_n(1)=\frac{1\cdot(-1)!!(n-2)!!}{n(n-2)!!}
=\frac1n.
\]
Both parity classes have maximum \(1/n\), so the overall nontrivial maximum is \(1/n\).
\end{proof}

\begin{theorem}[BM-Cert spectral gap]\label{thm:spectral-gap}
For \(n\ge3\), the BM-Cert verification operator \(\Omega_{\BM}\) has perfect completeness for \(\ket{G_n}\). Its second eigenvalue is
\[
\beta_{\BM}(n)=
\begin{cases}
\dfrac{1}{n-1}, & n\ge4\textnormal{ even},\\[8pt]
\dfrac{1}{n}, & n\ge3\textnormal{ odd},
\end{cases}
\]
and its spectral gap is
\[
\nu_{\BM}(n)=1-\beta_{\BM}(n)=
\begin{cases}
1-\dfrac{1}{n-1}, & n\ge4\textnormal{ even},\\[8pt]
1-\dfrac{1}{n}, & n\ge3\textnormal{ odd}.
\end{cases}
\]
\end{theorem}

\begin{proof}
Perfect completeness follows from Lemma~\ref{lem:perfect-completeness}. By Lemma~\ref{lem:cut-prob}, the GHZ-basis vectors \(\ket{a,\tau}\) diagonalize \(\Omega_{\BM}\). The target vector \(\ket{0^n,+}=\ket{G_n}\) has eigenvalue \(1\), because the empty cut is always avoided. The vector \(\ket{0^n,-}\) has eigenvalue \(0\), because \(\tau=-1\).

All other basis vectors correspond to nontrivial cuts: \(1\le |A(a)|\le n-1\). For \(\tau=-1\), Lemma~\ref{lem:cut-prob} gives eigenvalue \(0\). For \(\tau=+1\), Lemma~\ref{lem:cut-prob} says that the eigenvalue is the cut-avoidance probability. If \(n\) is even, Lemma~\ref{lem:even-count} shows that the largest nontrivial cut-avoidance probability is \(1/(n-1)\). If \(n\) is odd, Lemma~\ref{lem:odd-count} shows that it is \(1/n\). Therefore the largest eigenvalue on the subspace orthogonal to \(\ket{G_n}\) is the claimed \(\beta_{\BM}(n)\).
\end{proof}

\begin{corollary}[Copy complexity of BM-Cert]
\label{cor:sample-complexity}
Let \(\eps,\delta\in(0,1)\). Suppose BM-Cert is run for
\[
N_{\BM}(n,\eps,\delta)=
\left\lceil
\frac{\ln(1/\delta)}{-\ln(1-\nu_{\BM}(n)\eps)}
\right\rceil
\]
independent rounds, accepting only if all rounds pass. If the average infidelity of the tested copies with \(\ket{G_n}\) is at least \(\eps\), then the probability of accepting is at most \(\delta\). If every tested copy is exactly \(\ket{G_n}\), then the probability of accepting is \(1\).
\end{corollary}

\section{Optimality inside the Bell-matching family and comparison with Pauli verification}

We show that the spectral gap in Theorem~\ref{thm:spectral-gap} cannot be improved without leaving the stated measurement family or allowing completeness error.

\begin{definition}
\label{def:family}
A perfect-completeness Bell-matching strategy is any randomized single-copy strategy of the following kind.
\begin{enumerate}[label=(\roman*), leftmargin=2em]
\item It chooses any quasi-perfect matching \(Q\) according to any probability distribution.
\item It performs Bell-basis measurements on all pairs in \(M(Q)\), and, when \(n\) is odd, an \(X\)-basis measurement on the singleton.
\item It accepts or rejects using any classical randomized or deterministic function of the outcomes.
\item The target \(\ket{G_n}\) is accepted with probability \(1\).
\end{enumerate}
\end{definition}

\begin{lemma}
\label{lem:forced-outcomes}
Fix a quasi-perfect matching \(Q\). In the measurement associated with \(Q\), every outcome accepted by BM-Cert occurs with positive probability when the input state is \(\ket{G_n}\). Therefore, any perfect-completeness postprocessing rule for this fixed measurement must accept all BM-Cert pass outcomes. Consequently, its pass effect \(E_Q\) satisfies the operator inequality
\[
E_Q\ge \Pi_Q.
\]
\end{lemma}

\begin{proof}
First let \(n=2m\) be even. For each pair, write
\[
\ket{00}=\frac{\ket{\Phi^+}+\ket{\Phi^-}}{\sqrt2},
\qquad
\ket{11}=\frac{\ket{\Phi^+}-\ket{\Phi^-}}{\sqrt2}.
\]
For a sign vector \(r=(r_1,\ldots,r_m)\in\{+1,-1\}^m\), let \(\ket{\Phi^{r_j}}\) mean \(\ket{\Phi^+}\) if \(r_j=+1\) and \(\ket{\Phi^-}\) if \(r_j=-1\). Expanding \(\ket{0^n}\) and \(\ket{1^n}\) across the matched pairs gives
\[
\ket{0^n}=2^{-m/2}\sum_{r\in\{\pm1\}^m}\bigotimes_{j=1}^m\ket{\Phi^{r_j}},
\]
\[
\ket{1^n}=2^{-m/2}\sum_{r\in\{\pm1\}^m}\left(\prod_{j=1}^m r_j\right)\bigotimes_{j=1}^m\ket{\Phi^{r_j}}.
\]
Therefore
\[
\ket{G_n}
=\frac{\ket{0^n}+\ket{1^n}}{\sqrt2}
=2^{(1-m)/2}
\sum_{\substack{r\in\{\pm1\}^m\\ \prod_jr_j=+1}}
\bigotimes_{j=1}^m\ket{\Phi^{r_j}}.
\]
This shows that exactly the outcomes with all \(Z_iZ_j\) signs equal to \(+1\) and total \(X_iX_j\) sign product \(+1\) occur, and each occurs with probability \(2^{1-m}>0\). These are the BM-Cert pass outcomes.

Now let \(n=2m+1\) be odd, and let the singleton be \(s\). Use the additional identities
\[
\ket0=\frac{\ket{+}+\ket{-}}{\sqrt2},
\qquad
\ket1=\frac{\ket{+}-\ket{-}}{\sqrt2},
\]
where \(X\ket{\pm}=\pm\ket{\pm}\). The same expansion gives
\[
\ket{G_n}
=2^{-m/2}
\sum_{\substack{r_0,r_1,\ldots,r_m\in\{\pm1\}\\ r_0\prod_{j=1}^m r_j=+1}}
\ket{r_0}_s\otimes\bigotimes_{j=1}^m\ket{\Phi^{r_j}},
\]
where \(\ket{r_0}_s\) means \(\ket+\) if \(r_0=+1\) and \(\ket-\) if \(r_0=-1\). Thus exactly the outcomes with all pairwise \(Z_iZ_j\) signs \(+1\) and total \(X\)-parity \(+1\) occur, each with probability \(2^{-m}>0\). Again these are the BM-Cert pass outcomes.

If a strategy used this measurement with positive probability and its postprocessing rule rejected any outcome that occurs with positive probability on \(\ket{G_n}\), then the conditional acceptance probability for \(\ket{G_n}\) under this measurement would be less than \(1\), contradicting perfect completeness of the strategy. Therefore all BM-Cert pass outcomes must be accepted. The pass effect is a sum of orthogonal outcome projectors. Since it includes every outcome projector appearing in \(\Pi_Q\), it satisfies \(E_Q\ge\Pi_Q\).
\end{proof}

\begin{theorem}
\label{thm:optimal-family}
Among all perfect-completeness Bell-matching strategies, the minimum possible second eigenvalue is
\[
\beta_{\min}(n)=
\begin{cases}
\dfrac{1}{n-1}, & n\ge4\textnormal{ even},\\[8pt]
\dfrac{1}{n}, & n\ge3\textnormal{ odd}.
\end{cases}
\]
Equivalently, the maximum possible spectral gap in this family is \(\nu_{\BM}(n)\) from Theorem~\ref{thm:spectral-gap}. Uniform BM-Cert attains this optimum.
\end{theorem}

\begin{proof}
Theorem~\ref{thm:spectral-gap} proves that uniform BM-Cert attains the second eigenvalue \(\beta_{\min}(n)\). It remains to prove that no perfect-completeness Bell-matching strategy can have a smaller second eigenvalue.

First let \(n\ge4\) be even. Consider an arbitrary strategy in the family. It chooses a perfect matching \(Q\) from some distribution and then uses a perfect-completeness postprocessing rule. Let \(\Omega\) be its verification operator. For any unordered pair \(\{i,j\}\), let \(a^{ij}\in\{0,1\}^n\) be the bit string with ones only at positions \(i,j\). The vector \(\ket{a^{ij},+}\) is orthogonal to \(\ket{G_n}\).

For a fixed matching \(Q\), the minimal BM-Cert projector \(\Pi_Q\) accepts \(
\ket{a^{ij},+}\) if and only if \(\{i,j\}\in M(Q)\). Indeed, if \(i\) and \(j\) are paired together, then that pair has equal \(a^{ij}\)-bits, both equal to \(1\), and every other pair has equal \(a^{ij}\)-bits, both equal to \(0\). If \(i\) and \(j\) are not paired together, then the qubit \(i\) is paired with a zero bit or the qubit \(j\) is paired with a zero bit, so some measured \(Z\)-parity is \(-1\). Thus
\[
\bra{a^{ij},+}\Pi_Q\ket{a^{ij},+}=\mathbf 1\{\{i,j\}\in M(Q)\}.
\]
By Lemma~\ref{lem:forced-outcomes}, the actual pass effect for matching \(Q\) is at least \(\Pi_Q\). Therefore
\[
\bra{a^{ij},+}\Omega\ket{a^{ij},+}
\ge \Pr_Q[\{i,j\}\in M(Q)].
\]
Average this inequality over all \(\binom n2\) unordered pairs:
\[
\frac{1}{\binom n2}\sum_{1\le i<j\le n}
\bra{a^{ij},+}\Omega\ket{a^{ij},+}
\ge
\frac{1}{\binom n2}\mathbb E_Q |M(Q)|.
\]
Every perfect matching has \(|M(Q)|=n/2\), so
\[
\frac{1}{\binom n2}\mathbb E_Q |M(Q)|
=\frac{n/2}{n(n-1)/2}
=\frac{1}{n-1}.
\]
Hence at least one pair \(\{i,j\}\) satisfies
\[
\bra{a^{ij},+}\Omega\ket{a^{ij},+}
\ge\frac{1}{n-1}.
\]
Since \(\ket{a^{ij},+}\) is a unit vector orthogonal to \(\ket{G_n}\), the largest eigenvalue of \(\Omega\) on the orthogonal subspace is at least this expectation value. Thus \(\beta(\Omega)\ge1/(n-1)\).

Now let \(n\ge3\) be odd. For each vertex \(i\), let \(a^i\in\{0,1\}^n\) be the bit string with a single one at position \(i\). The vector \(\ket{a^i,+}\) is orthogonal to \(\ket{G_n}\). For a fixed near-perfect matching \(Q\), the minimal projector \(\Pi_Q\) accepts \(\ket{a^i,+}\) if and only if \(i\) is the singleton. If \(i\) is the singleton, then all paired vertices have zero \(a^i\)-bits, so every pairwise \(Z\)-parity is \(+1\), and the global \(X\)-parity condition is satisfied because the state has \(\tau=+1\). If \(i\) is paired with any other vertex, that pair crosses the cut and has \(Z\)-parity \(-1\). Hence
\[
\bra{a^i,+}\Pi_Q\ket{a^i,+}=\mathbf 1\{s(Q)=i\}.
\]
Again Lemma~\ref{lem:forced-outcomes} implies
\[
\bra{a^i,+}\Omega\ket{a^i,+}
\ge \Pr_Q[s(Q)=i].
\]
Averaging over \(i\in[n]\),
\[
\frac1n\sum_{i=1}^n \bra{a^i,+}\Omega\ket{a^i,+}
\ge\frac1n\sum_{i=1}^n\Pr_Q[s(Q)=i]
=\frac1n,
\]
because every near-perfect matching has one singleton. Therefore some \(i\) has expectation at least \(1/n\), and since \(\ket{a^i,+}\perp\ket{G_n}\), we get \(\beta(\Omega)\ge1/n\).
\end{proof}

\begin{table}[t]
\caption{\justifying Comparison of the spectral gaps of BM-Cert with standard verification benchmarks for the \(n\)-qubit GHZ state.}
\label{tab:ghz-comparison}
\centering
\setlength{\tabcolsep}{12pt}
\begin{tabular}{@{}lll@{}}
\toprule
Measurement model & Second eigenvalue \(\beta\) & Spectral gap \(\nu=1-\beta\) \\
\midrule
Unrestricted projector & \(0\) & \(1\) \\
Optimal local Pauli \cite{li2020optimal} & \(1/3\) & \(2/3\) \\
BM-Cert, even \(n\ge4\) & \(1/(n-1)\) & \(1-1/(n-1)\) \\
BM-Cert, odd \(n\ge3\) & \(1/n\) & \(1-1/n\) \\
\bottomrule
\end{tabular}
\end{table}

The GHZ state has an optimal local Pauli verification strategy with spectral gap \(2/3\), as shown in the GHZ-specific optimal verification work of Li et al. \cite{li2020optimal}. The same value is also the upper bound for separable measurements on entangled stabilizer states in the Pauli verification analysis of Dangniam et al. \cite{dangniam2020optimal}. Table~\ref{tab:ghz-comparison} compares this benchmark with BM-Cert.

\noindent\textbf{Copy-complexity improvement over local Pauli GHZ verification.}
For every odd \(n\ge5\) and every even \(n\ge6\), BM-Cert has strictly larger spectral gap than the optimal local Pauli GHZ verifier:
\[
\nu_{\BM}(n)>\frac23.
\]
Consequently, in the high-precision regime \(\eps\to0\) with fixed \(\delta\), the leading copy coefficient is strictly smaller for BM-Cert. More precisely,
\[
\frac{N_{\BM}(n,\eps,\delta)}{N_{\mathrm{Pauli}}(\eps,\delta)}
\longrightarrow
\begin{cases}
\dfrac{2(n-1)}{3(n-2)}, & n\ge6\text{ even},\\[10pt]
\dfrac{2n}{3(n-1)}, & n\ge5\text{ odd},
\end{cases}
\]
as \(\eps\to0\), where \(N_{\mathrm{Pauli}}(\eps,\delta)\) denotes the corresponding bound with spectral gap \(2/3\). To see this, for any fixed spectral gap \(\nu\), Lemma~\ref{lem:many-copy} gives
\[
N(\nu)=
\left\lceil
\frac{\ln(1/\delta)}{-\ln(1-\nu\eps)}
\right\rceil .
\]
Using
\[
-\ln(1-\nu\eps)=\nu\eps+O(\eps^2)
\]
as \(\eps\to0\), the unrounded bound has leading term
\[
\frac{\ln(1/\delta)}{\nu\eps}.
\]
Substituting \(\nu_{\BM}(n)=(n-2)/(n-1)\) for even \(n\) gives
\[
\frac{N_{\BM}(n,\eps,\delta)}{N_{\mathrm{Pauli}}(\eps,\delta)}
\longrightarrow
\frac{2/3}{(n-2)/(n-1)}
=
\frac{2(n-1)}{3(n-2)}.
\]
Substituting \(\nu_{\BM}(n)=(n-1)/n\) for odd \(n\) gives
\[
\frac{N_{\BM}(n,\eps,\delta)}{N_{\mathrm{Pauli}}(\eps,\delta)}
\longrightarrow
\frac{2/3}{(n-1)/n}
=
\frac{2n}{3(n-1)}.
\]
Both ratios tend to \(2/3\) as \(n\to\infty\).

\section{Limits of extending Algorithm~\ref{alg:bm-cert} to general stabilizer states}
\label{sec:counterexample-arbitrary-stabilizer}

It is natural to ask whether the near-projective verification developed for GHZ states can be extended to arbitrary stabilizer targets using only 2-local measurements. We show that this is impossible, even if one allows arbitrary disjoint 2-local POVMs.

Suppose one chooses a partition \(\mathcal B\) of \([n]\) into blocks of size \(1\) or \(2\).  On each block \(B\in\mathcal B\), one performs an arbitrary POVM \(\{M^B_a\}_a\).  The verifier then accepts with an arbitrary probability \(q(a_B:B\in\mathcal B)\in[0,1]\) depending on the classical outcome string.  The resulting pass effect has the form
\[
E=\sum_{(a_B)} q(a_B:B\in\mathcal B)\bigotimes_{B\in\mathcal B}M^B_{a_B},
\]
and satisfies \(0\le E\le\id\).  A general randomized strategy samples such a round \(r\) with probability \(p_r\), and its verification operator is
\[
\Omega=\sum_r p_rE_r.
\]

For \(n\ge3\), define the stabilizer state
\[
\ket{\Psi_n}:=\ket{G_3}_{123}\otimes\ket{0^{n-3}}_{4,\ldots,n}.
\]

\begin{lemma}
\label{lem:perfect-completeness-round-by-round}
Let \(\Omega=\sum_r p_rE_r\) be a randomized verification operator with \(0\le E_r\le\id\) for every \(r\).  If \(\Omega\) has perfect completeness for a target \(\ket\phi\), then every round used with positive probability has perfect completeness for \(\ket\phi\):
\[
E_r\ket\phi=\ket\phi
\qquad\textnormal{whenever }p_r>0.
\]
\end{lemma}

\begin{proof}
Perfect completeness gives
$
1=\bra\phi\Omega\ket\phi=\sum_r p_r\bra\phi E_r\ket\phi .
$
Since \(0\le E_r\le \id\), each term \(\bra\phi E_r\ket\phi\) is at most \(1\). Hence, if \(p_r>0\), then \(\bra\phi E_r\ket\phi=1\). Therefore
$
0=\bra\phi(\id-E_r)\ket\phi
=\bigl\|(\id-E_r)^{1/2}\ket\phi\bigr\|^2 .
$
Thus \((\id-E_r)^{1/2}\ket\phi=0\), and hence \(E_r\ket\phi=\ket\phi\).
\end{proof}

\begin{lemma}
\label{lem:compression-first-three}
Let \(E_r\) be the pass effect of a round whose measurements act on blocks of size at most 2 and that has perfect completeness for \(\ket{\Psi_n}\).
Define the compressed three-qubit operator
\[
F_r:=\bigl(\id_{123}\otimes\bra{0^{n-3}}\bigr)
      E_r
      \bigl(\id_{123}\otimes\ket{0^{n-3}}\bigr),
\]
with the evident convention when \(n=3\).  Then
\[
0\le F_r\le\id_{123},
\qquad
F_r\ket{G_3}=\ket{G_3}.
\]
Moreover, \(F_r\) is separable across at least one of the three bipartitions
\[
1|23,
\qquad
2|13,
\qquad
3|12.
\]
\end{lemma}

\begin{proof}
The inequalities \(0\le F_r\le\id_{123}\) follow by taking the matrix element of \(0\le E_r\le\id\) against the spectator state \(\ket{0^{n-3}}\).  Since \(E_r\ket{\Psi_n}=\ket{\Psi_n}\), we have
\[
\begin{aligned}
F_r\ket{G_3}
&=\bigl(\id_{123}\otimes\bra{0^{n-3}}\bigr)
  E_r
  \bigl(\ket{G_3}_{123}\otimes\ket{0^{n-3}}_{4,\ldots,n}\bigr) \\
&=\bigl(\id_{123}\otimes\bra{0^{n-3}}\bigr)
  \bigl(\ket{G_3}_{123}\otimes\ket{0^{n-3}}_{4,\ldots,n}\bigr) \\
&=\ket{G_3}.
\end{aligned}
\]

For the fixed round \(r\), write its block partition as \(\mathcal B_r\).  Its pass effect is a positive linear combination of tensor products over the blocks:
\[
E_r=\sum_{(a_B)}q_r(a_B:B\in\mathcal B_r)
       \bigotimes_{B\in\mathcal B_r}M^{B,r}_{a_B}.
\]
Compress one tensor-product term against \(\ket{0^{n-3}}\).  A block \(B\) disjoint from \(\{1,2,3\}\) contributes a nonnegative scalar.  A block meeting \(\{1,2,3\}\) in one qubit contributes a positive operator on that one qubit.  A block meeting \(\{1,2,3\}\) in two qubits contributes a positive operator on that pair.  No block can meet \(\{1,2,3\}\) in all three qubits, because every block has size at most \(2\).

Thus each compressed tensor-product term is a tensor product of positive operators over an induced partition of \(\{1,2,3\}\) whose blocks have size at most \(2\).  Since the original block partition \(\mathcal B_r\) is fixed throughout the round, this induced partition is the same for all outcome strings in the sum.  Summing the compressed positive tensor-product terms preserves separability across that induced bipartition.  Hence \(F_r\) is separable across at least one of the bipartitions $1|23$, $2|13$, and $3|12$.
\end{proof}

\begin{lemma}
\label{lem:ghz3-trace-lower-bound}
Let \(F\) be a positive semidefinite operator on three qubits.  Suppose that \(F\) is separable across at least one of the bipartitions \(1|23\), \(2|13\), and \(3|12\).  If
$
\bra{G_3}F\ket{G_3}=1,
$
then
$
\Tr(F)\ge2.
$
\end{lemma}

\begin{proof}
By symmetry, it is enough to consider the case where \(F\) is separable across \(1|23\).  Across this cut,
\[
\ket{G_3}=\frac{\ket0\ket{00}+\ket1\ket{11}}{\sqrt2}
\]
has Schmidt coefficients \(1/\sqrt2\) and \(1/\sqrt2\).  Hence every product unit vector \(\ket a\otimes\ket b\in\mathbb C^2\otimes\mathbb C^4\) satisfies
\[
\left|\bigl(\bra a\otimes\bra b\bigr)\ket{G_3}\right|^2\le\frac12.
\]

Since \(F\) is positive and separable across \(1|23\), it can be written as a finite conic combination of rank-one product projectors,
\[
F=\sum_\ell c_\ell\,\proj{a_\ell}\otimes\proj{b_\ell},
\qquad
c_\ell\ge0,
\]
where \(\ket{a_\ell}\) and \(\ket{b_\ell}\) are unit vectors on the first qubit and on the last two qubits, respectively.  Therefore
\[
\bra{G_3}F\ket{G_3}
=\sum_\ell c_\ell
\left|\bigl(\bra{a_\ell}\otimes\bra{b_\ell}\bigr)\ket{G_3}\right|^2
\le\frac12\sum_\ell c_\ell
=\frac12\Tr(F).
\]
The assumption \(\bra{G_3}F\ket{G_3}=1\) gives \(1\le\Tr(F)/2\), and hence \(\Tr(F)\ge2\).
\end{proof}

\begin{theorem}
\label{thm:counterexample-two-local}
For every \(n\ge3\), let
\[
\ket{\Psi_n}=\ket{G_3}_{123}\otimes\ket{0^{n-3}}_{4,\ldots,n}.
\]
Let \(\Omega\) be the verification operator of any perfect-completeness strategy for \(\ket{\Psi_n}\) whose rounds consist of disjoint measurements on blocks of size at most 2.  Then
\[
\beta(\Omega)\ge\frac17,
\qquad
\nu(\Omega)\le\frac67.
\]
\end{theorem}

\begin{proof}
Write the randomized strategy as
\[
\Omega=\sum_r p_rE_r,
\]
where each \(E_r\) is the pass effect of one 2-local round.  By Lemma~\ref{lem:perfect-completeness-round-by-round}, every \(E_r\) with \(p_r>0\) satisfies
\[
E_r\ket{\Psi_n}=\ket{\Psi_n}.
\]
For such a round, let \(F_r\) be the three-qubit compression from Lemma~\ref{lem:compression-first-three}.  Then
\[
0\le F_r\le\id_{123},
\qquad
F_r\ket{G_3}=\ket{G_3},
\]
and \(F_r\) is separable across at least one of the one-versus-two bipartitions.  In particular,
\[
\bra{G_3}F_r\ket{G_3}=1.
\]
Lemma~\ref{lem:ghz3-trace-lower-bound} therefore gives
\[
\Tr(F_r)\ge2
\qquad\textnormal{for every }r\textnormal{ with }p_r>0.
\]

Define the averaged compressed operator
\[
\overline F:=\sum_r p_rF_r
=\bigl(\id_{123}\otimes\bra{0^{n-3}}\bigr)
  \Omega
  \bigl(\id_{123}\otimes\ket{0^{n-3}}\bigr).
\]
Then
\[
0\le\overline F\le\id_{123},
\qquad
\overline F\ket{G_3}=\ket{G_3},
\qquad
\Tr(\overline F)\ge2.
\]
Let \(P=\proj{G_3}\) and \(P^\perp=\id_{123}-P\).  Since \(\overline F\ket{G_3}=\ket{G_3}\), the vector \(\ket{G_3}\) is an eigenvector of \(\overline F\) with eigenvalue \(1\), and
\[
\Tr(P^\perp\overline F P^\perp)=\Tr(\overline F)-1\ge1.
\]
The operator \(P^\perp\overline F P^\perp\) is positive semidefinite on the seven-dimensional subspace orthogonal to \(\ket{G_3}\).  Therefore its largest eigenvalue is at least its average eigenvalue:
\[
\bigl\|P^\perp\overline F P^\perp\bigr\|
\ge\frac{\Tr(P^\perp\overline F P^\perp)}{7}
\ge\frac17.
\]
Thus the compressed three-qubit operator has second eigenvalue at least \(1/7\).

Finally, let
\[
V:(\mathbb C^2)^{\otimes3}\to(\mathbb C^2)^{\otimes n},
\qquad
V\ket\varphi=\ket\varphi_{123}\otimes\ket{0^{n-3}}_{4,\ldots,n}.
\]
Then \(V\) is an isometry, \(V\ket{G_3}=\ket{\Psi_n}\), and \(\overline F=V^\dagger\Omega V\).  For every unit vector \(\ket\varphi\perp\ket{G_3}\), the vector \(V\ket\varphi\) is a unit vector orthogonal to \(\ket{\Psi_n}\), and
\[
\bra\varphi\overline F\ket\varphi
=\bra\varphi V^\dagger\Omega V\ket\varphi
=\bra{V\varphi}\Omega\ket{V\varphi}.
\]
Taking the maximum over \(\ket\varphi\perp\ket{G_3}\) gives
\[
\beta(\Omega)
=\bigl\|(\id-\proj{\Psi_n})\Omega(\id-\proj{\Psi_n})\bigr\|
\ge
\bigl\|P^\perp\overline F P^\perp\bigr\|
\ge\frac17.
\]
Since \(\nu(\Omega)=1-\beta(\Omega)\), this implies \(\nu(\Omega)\le6/7\).
\end{proof}

\section{Fidelity estimation from Bell-matching outcomes}
\label{sec:bm-fid-algorithm}

\begin{algorithm}[!htbp]
\caption{\justifying BM-Fid for \(\ket{G_n}\), with \(n\ge3\)}
\label{alg:bm-fid}
\begin{algorithmic}[1]
\Require Estimation accuracy $\eps\in(0,1)$ and failure probability $\delta\in(0,1)$

\State Set $m=\lfloor n/2\rfloor$ and
\[
R_{\BM}^{\mathrm{fid}}(n)=
\begin{cases}
1+\dfrac{1}{n-2}, & n \text{ even},\\[8pt]
1+\dfrac{1}{n-1}, & n \text{ odd},
\end{cases}
\]
\Statex and
\[
N=
\left\lceil
\frac{\bigl(R_{\BM}^{\mathrm{fid}}(n)\bigr)^2\ln(2/\delta)}{2\eps^2}
\right\rceil .
\]

\For{each round $t=1,\ldots,N$}
    \State Sample a uniformly random quasi-perfect matching $Q_t$ of $[n]$

    \For{every pair $\{i,j\}\in M(Q_t)$}
        \State Perform a Bell-basis measurement on qubits $i,j$
        \State Record its outcome $z_{ij},x_{ij}\in\{+1,-1\}$
    \EndFor

    \If{$n$ is odd}
        \State Measure the singleton qubit $s(Q_t)$ in the $X$ basis
        \State Record its outcome $x_{s(Q_t)}\in\{+1,-1\}$
    \EndIf

    \State Set
    \Statex
    \[
    \xi_t=
    \begin{cases}
    \displaystyle\prod_{\{i,j\}\in M(Q_t)}x_{ij}, & n \text{ even},\\[12pt]
    \displaystyle x_{s(Q_t)}\prod_{\{i,j\}\in M(Q_t)}x_{ij}, & n \text{ odd}.
    \end{cases}
    \]

    \State Set the single-round fidelity score
    \Statex
    \[
    Y_t=
    \frac{1+\xi_t}{2^n}
    \sum_{T\subseteq M(Q_t)}
    \frac{\binom n{2|T|}}{\binom m{|T|}}
    \prod_{\{i,j\}\in T}z_{ij}.
    \]
\EndFor

\Statex \textbf{Final output:} Return the estimate
$
\widehat F_{\BM}^{\mathrm{fid}}:=\frac1N\sum_{t=1}^NY_t .
$
\end{algorithmic}
\end{algorithm}

Algorithm~\ref{alg:bm-cert} uses the Bell-matching outcomes only through a binary accept/reject rule.  We now use the same measurement outcomes in Algorithm~\ref{alg:bm-fid} to estimate the GHZ fidelity
\[
F_{G_n}(\rho)=\bra{G_n}\rho\ket{G_n}.
\]
For a subset $U\subseteq[n]$, write
\[
Z_U:=\prod_{i\in U}Z_i,
\qquad
Z_\emptyset:=\id .
\]
For a quasi-perfect matching $Q$ and a subset $T\subseteq M(Q)$ of its matched pairs, write
\[
V_Q(T):=\bigcup_{\{i,j\}\in T}\{i,j\}.
\]
Thus $V_Q(T)$ is the set of all vertices covered by the pairs in $T$.  We also use the convention that an empty product of signs is equal to $1$.

\begin{lemma}
\label{lem:bm-fid-ghz-projector-expansion}
For every $n\ge2$,
\[
\proj{G_n}
=
\frac1{2^n}
\sum_{\substack{U\subseteq[n]\\ |U|\,\textnormal{even}}}
\left(Z_U+X^{\otimes n}Z_U\right).
\]
\end{lemma}

\begin{proof}
Let
\[
D:=\sum_{\substack{U\subseteq[n]\\ |U|\,\textnormal{even}}} Z_U =\frac12\left(\prod_{i=1}^n(\id+Z_i)+\prod_{i=1}^n(\id-Z_i)\right).
\]
Now
\[
\prod_{i=1}^n(\id+Z_i)=2^n\proj{0^n},
\qquad
\prod_{i=1}^n(\id-Z_i)=2^n\proj{1^n},
\]
and therefore
\[
D=2^{n-1}\left(\proj{0^n}+\proj{1^n}\right).
\]
Multiplying on the left by $X^{\otimes n}$ gives
\[
X^{\otimes n}D
=2^{n-1}\left(\ket{1^n}\!\bra{0^n}+\ket{0^n}\!\bra{1^n}\right).
\]
Hence
\[
\begin{aligned}
\frac1{2^n}\sum_{\substack{U\subseteq[n]\\ |U|\,\textnormal{even}}}
\left(Z_U+X^{\otimes n}Z_U\right)
&=\frac1{2^n}\left(D+X^{\otimes n}D\right)\\
&=\frac12\left(\proj{0^n}+\proj{1^n}
+\ket{1^n}\!\bra{0^n}+\ket{0^n}\!\bra{1^n}\right)\\
&=\proj{G_n}.
\end{aligned}
\]
\end{proof}

\begin{lemma}
\label{lem:bm-fid-inclusion-probability}
Let $Q$ be a uniformly random quasi-perfect matching of $[n]$, and put $m=\lfloor n/2\rfloor$.  Fix an even subset $U\subseteq[n]$ with $|U|=2r$, where $0\le r\le m$.  Then
\[
\Pr_Q\left[U=V_Q(T)\textnormal{ for some }T\subseteq M(Q)\right]
=
\frac{\binom m r}{\binom n{2r}}.
\]
\end{lemma}

\begin{proof}
For each fixed quasi-perfect matching $Q$, there are exactly $\binom m r$ subsets $T\subseteq M(Q)$ with $|T|=r$.  Each such $T$ determines one subset $V_Q(T)\subseteq[n]$ of size $2r$.  Hence the average number of size-$2r$ subsets represented in the form $V_Q(T)$ is $\binom m r$.

By symmetry, every fixed subset $U\subseteq[n]$ of size $2r$ has the same probability of being represented.  There are $\binom n{2r}$ such subsets.  Moreover, if $U$ is represented, the corresponding $T$ is unique.  Therefore
\[
\binom n{2r}\Pr_Q\left[U=V_Q(T)\textnormal{ for some }T\subseteq M(Q)\right]
=\binom m r.
\]
\end{proof}

\begin{theorem}
\label{thm:bm-fid-unbiasedness}
Let $\rho$ be any $n$-qubit state, and let $Y$ be the single-round score produced by Algorithm~\ref{alg:bm-fid} on input $\rho$.  Then
\[
\mathbb E[Y]=F_{G_n}(\rho)=\bra{G_n}\rho\ket{G_n}.
\]
More generally, if the tested copies have states $\rho_1,\ldots,\rho_N$, then
\[
\mathbb E\left[\widehat F_{\BM}^{\mathrm{fid}}\right]
=\frac1N\sum_{t=1}^NF_{G_n}(\rho_t).
\]
\end{theorem}

\begin{proof}
Fix a quasi-perfect matching $Q$.  The Bell measurement on a pair $\{i,j\}$ jointly measures $Z_iZ_j$ and $X_iX_j$.  Therefore, for each $T\subseteq M(Q)$, the product
\[
\prod_{\{i,j\}\in T}z_{ij}
\]
is the measured eigenvalue of $Z_{V_Q(T)}$.  The sign
\[
\xi=
\begin{cases}
\displaystyle\prod_{\{i,j\}\in M(Q)}x_{ij}, & n \textnormal{ even},\\[12pt]
\displaystyle x_{s(Q)}\prod_{\{i,j\}\in M(Q)}x_{ij}, & n \textnormal{ odd},
\end{cases}
\]
is the measured eigenvalue of $X^{\otimes n}$.  Since $|V_Q(T)|$ is even, $Z_{V_Q(T)}$ commutes with $X^{\otimes n}$.  Thus
\[
\xi\prod_{\{i,j\}\in T}z_{ij}
\]
is the measured eigenvalue of $X^{\otimes n}Z_{V_Q(T)}$.

Consequently, conditional on the choice of $Q$,
\[
\mathbb E[Y\mid Q]
=\Tr(H_Q\rho),
\]
where
\[
H_Q:=
\frac1{2^n}
\sum_{T\subseteq M(Q)}
\frac{\binom n{2|T|}}{\binom m{|T|}}
\left(Z_{V_Q(T)}+X^{\otimes n}Z_{V_Q(T)}\right).
\]
Averaging $H_Q$ over a uniformly random quasi-perfect matching and using Lemma~\ref{lem:bm-fid-inclusion-probability}, each fixed even subset $U\subseteq[n]$ with $|U|=2r$ receives total coefficient
\[
\frac{\binom n{2r}}{\binom m r}
\Pr_Q\left[U=V_Q(T)\textnormal{ for some }T\subseteq M(Q)\right]
=1.
\]
Therefore
\[
\mathbb E_Q H_Q
=\frac1{2^n}
\sum_{\substack{U\subseteq[n]\\ |U|\,\textnormal{even}}}
\left(Z_U+X^{\otimes n}Z_U\right)
=\proj{G_n},
\]
where the last equality follows from Lemma~\ref{lem:bm-fid-ghz-projector-expansion}.  Hence
\[
\mathbb E[Y]
=\Tr\left((\mathbb E_QH_Q)\rho\right)
=\Tr(\proj{G_n}\rho)
=F_{G_n}(\rho).
\]
Applying this separately to the $N$ tested copies gives
\[
\mathbb E\left[\widehat F_{\BM}^{\mathrm{fid}}\right]
=\frac1N\sum_{t=1}^N\mathbb E[Y_t]
=\frac1N\sum_{t=1}^NF_{G_n}(\rho_t).
\]
\end{proof}

\begin{lemma}
\label{lem:bm-fid-score-range}
For every possible outcome of one round of Algorithm~\ref{alg:bm-fid},
\[
-\frac1{n-2}\le Y\le1
\qquad
(n\ge4\textnormal{ even}),
\]
and
\[
-\frac1{n-1}\le Y\le1
\qquad
(n\ge3\textnormal{ odd}).
\]
\end{lemma}

\begin{proof}
If the measured global $X$-parity sign is $\xi=-1$, then $Y=0$.  It remains to consider the case $\xi=+1$.  Let $s$ be the number of pairs $\{i,j\}\in M(Q)$ for which $z_{ij}=-1$.  For $0\le r\le m$, define
\[
K_r^{(m)}(s):=\sum_{j=0}^r(-1)^j\binom s j\binom{m-s}{r-j},
\]
with the convention that $\binom ab=0$ when $b<0$ or $b>a$.  This is the sum of $\prod_{\{i,j\}\in T}z_{ij}$ over all $r$-element subsets $T\subseteq M(Q)$, because $j$ counts how many of the chosen pairs have $z_{ij}=-1$.  Hence, when $\xi=+1$,
\[
Y=B_n(s),
\]
where
\[
B_n(s):=
2^{1-n}
\sum_{r=0}^m
\frac{\binom n{2r}}{\binom m r}K_r^{(m)}(s).
\]

First let $n=2m$ be even.  Define
\[
G_{2m}(y):=\sum_{s=0}^m\binom m s B_{2m}(s)y^s .
\]
From the definition of $K_r^{(m)}(s)$,
\[
K_r^{(m)}(s)=[x^r](1-x)^s(1+x)^{m-s}.
\]
Therefore
\[
\begin{aligned}
\sum_{s=0}^m\binom m sK_r^{(m)}(s)y^s
&=[x^r]\sum_{s=0}^m\binom m s\bigl(y(1-x)\bigr)^s(1+x)^{m-s}\\
&=[x^r]\bigl((1+x)+y(1-x)\bigr)^m\\
&=[x^r]\bigl((1+y)+(1-y)x\bigr)^m\\
&=\binom m r(1-y)^r(1+y)^{m-r}.
\end{aligned}
\]
It follows that
\[
\begin{aligned}
G_{2m}(y)
&=2^{1-2m}\sum_{r=0}^m\binom{2m}{2r}(1-y)^r(1+y)^{m-r}\\
&=2^{-m}\left(\bigl(1+\sqrt{1-y^2}\bigr)^m
+\bigl(1-\sqrt{1-y^2}\bigr)^m\right).
\end{aligned}
\]
The last equality is obtained by setting $a=\sqrt{1-y}$ and $b=\sqrt{1+y}$ and using
\[
\sum_{r=0}^m\binom{2m}{2r}a^{2r}b^{2m-2r}
=\frac{(a+b)^{2m}+(b-a)^{2m}}2.
\]
Since $G_{2m}(y)$ is a polynomial of degree at most $m$, it suffices to extract coefficients up to degree $m$.  For $0\le 2\ell\le m$, the term $\bigl(1-\sqrt{1-y^2}\bigr)^m$ has no contribution to the coefficient of $y^{2\ell}$.  Put $z=y^2$.  We claim that
\[
[z^\ell]\bigl(1+\sqrt{1-z}\bigr)^m
=(-1)^\ell 2^{m-2\ell}\frac{m}{m-\ell}\binom{m-\ell}{\ell}.
\]
Indeed, by Cauchy's coefficient formula and the substitution $u=\sqrt{1-z}$,
\[
\begin{aligned}
[z^\ell]\bigl(1+\sqrt{1-z}\bigr)^m
&=\operatorname*{Res}_{z=0}\frac{(1+\sqrt{1-z})^m}{z^{\ell+1}}\,dz\\
&=\operatorname*{Res}_{u=1}\frac{(1+u)^m(-2u)}{(1-u^2)^{\ell+1}}\,du\\
&=(-1)^\ell\operatorname*{Res}_{u=1}
\frac{2u(1+u)^{m-\ell-1}}{(u-1)^{\ell+1}}\,du\\
&=(-1)^\ell[v^\ell]2(1+v)(2+v)^{m-\ell-1}\\
&=(-1)^\ell 2^{m-2\ell}\frac{m}{m-\ell}\binom{m-\ell}{\ell},
\end{aligned}
\]
where $v=u-1$ in the fourth line.  Thus
\[
B_{2m}(0)=1,
\]
\[
B_{2m}(s)=0\qquad\textnormal{for odd }s,
\]
and, for $\ell\ge1$ with $2\ell\le m$,
\[
B_{2m}(2\ell)
=(-1)^\ell
\frac{(2\ell-1)!!}{(2m-2)(2m-4)\cdots(2m-2\ell)}.
\]
If both $2\ell$ and $2\ell+2$ lie in $\{0,1,\ldots,m\}$, then for $\ell\ge1$,
\[
\frac{|B_{2m}(2\ell+2)|}{|B_{2m}(2\ell)|}
=\frac{2\ell+1}{2m-2\ell-2}\le1,
\]
because $2\ell+2\le m$ implies $2\ell+1\le2m-2\ell-2$.  Thus the negative term of largest magnitude is the first one, namely
\[
B_{2m}(2)=-\frac1{2m-2}=-\frac1{n-2}.
\]
All positive values are at most $B_{2m}(0)=1$.  Hence
\[
-\frac1{n-2}\le Y\le1
\]
for even $n$.

Now let $n=2m+1$ be odd.  Define
\[
G_{2m+1}(y):=\sum_{s=0}^m\binom m sB_{2m+1}(s)y^s .
\]
The same calculation gives
\[
G_{2m+1}(y)
=2^{-2m}\sum_{r=0}^m\binom{2m+1}{2r}(1-y)^r(1+y)^{m-r}.
\]
Let $H(y)$ denote the even generating function for $2m+2$ qubits, namely
\[
H(y):=G_{2m+2}(y)
=2^{-2m-1}\sum_{r=0}^{m+1}\binom{2m+2}{2r}(1-y)^r(1+y)^{m+1-r}.
\]
Regard the last sum temporarily as a homogeneous polynomial in two variables $A=1-y$ and $B=1+y$.  Since
\[
\binom{2m+1}{2r}=\frac{m+1-r}{m+1}\binom{2m+2}{2r},
\]
we have
\[
G_{2m+1}(y)=\frac{2}{m+1}\frac{\partial H}{\partial B}.
\]
On the line $A=1-y$, $B=1+y$, homogeneity gives
\[
A\frac{\partial H}{\partial A}+B\frac{\partial H}{\partial B}=(m+1)H,
\]
while differentiation with respect to $y$ gives
\[
H'(y)=-\frac{\partial H}{\partial A}+\frac{\partial H}{\partial B}.
\]
Solving these two equations and using $A+B=2$ yields
\[
G_{2m+1}(y)=H(y)+\frac{1-y}{m+1}H'(y).
\]
Since $H(y)$ has only even powers, write
\[
H(y)=\sum_{k=0}^{m+1}h_ky^k,
\qquad
h_k=\binom{m+1}{k}B_{2m+2}(k).
\]
Then $h_k=0$ for odd $k$.  From
$G_{2m+1}(y)=H(y)+(1-y)H'(y)/(m+1)$, the coefficient $g_s$ of $y^s$ in $G_{2m+1}(y)$ is
\[
g_s=\frac{m+1-s}{m+1}h_s+\frac{s+1}{m+1}h_{s+1}.
\]
For $s=2\ell-1$,
\[
g_{2\ell-1}=\frac{2\ell}{m+1}\binom{m+1}{2\ell}B_{2m+2}(2\ell)
=\binom m{2\ell-1}B_{2m+2}(2\ell).
\]
For $s=2\ell$,
\[
g_{2\ell}=\frac{m+1-2\ell}{m+1}\binom{m+1}{2\ell}B_{2m+2}(2\ell)
=\binom m{2\ell}B_{2m+2}(2\ell).
\]
Dividing by the corresponding binomial coefficients in the definition of $G_{2m+1}$ proves that, whenever the arguments lie between $0$ and $m$,
\[
B_{2m+1}(2\ell-1)=B_{2m+1}(2\ell)=B_{2m+2}(2\ell).
\]
Therefore
\[
B_{2m+1}(0)=1,
\]
and, for $\ell\ge1$,
\[
B_{2m+1}(2\ell-1)=B_{2m+1}(2\ell)
=(-1)^\ell
\frac{(2\ell-1)!!}{(2m)(2m-2)\cdots(2m-2\ell+2)},
\]
again whenever the arguments are in $\{0,1,\ldots,m\}$.  For subsequent nonzero values,
\[
\frac{|B_{2m+1}(2\ell+1)|}{|B_{2m+1}(2\ell-1)|}
=\frac{2\ell+1}{2m-2\ell}\le1
\]
whenever $2\ell+1\le m$.  Hence the negative term of largest magnitude is the first one, namely
\[
B_{2m+1}(1)=B_{2m+1}(2)=-\frac1{2m}=-\frac1{n-1},
\]
where the expression $B_{2m+1}(2)$ is ignored if $m=1$.  All positive values are at most $1$.  Hence
\[
-\frac1{n-1}\le Y\le1
\]
for odd $n$.
\end{proof}

\begin{lemma}[Hoeffding's inequality]
\label{lem:bm-fid-concentration}
Let $Y_1,\ldots,Y_N$ be independent real random variables with $a\le Y_t\le b$ for every $t$.  Let
\[
\widehat\mu:=\frac1N\sum_{t=1}^NY_t,
\qquad
\overline\mu:=\frac1N\sum_{t=1}^N\mathbb E[Y_t].
\]
Then, for every $\eps>0$,
\[
\Pr\left[|\widehat\mu-\overline\mu|\ge\eps\right]
\le
2\exp\left(-\frac{2N\eps^2}{(b-a)^2}\right).
\]
\end{lemma}

\begin{corollary}[Copy complexity of BM-Fid]
\label{cor:bm-fid-copy-complexity}
Let $n\ge3$ and $\eps,\delta\in(0,1)$.  Suppose Algorithm~\ref{alg:bm-fid} is run on independent copies with states $\rho_1,\ldots,\rho_N$, and define the average fidelity
\[
\overline F:=\frac1N\sum_{t=1}^NF_{G_n}(\rho_t).
\]
If
\[
N\ge
\left\lceil
\frac{\bigl(R_{\BM}^{\mathrm{fid}}(n)\bigr)^2\ln(2/\delta)}{2\eps^2}
\right\rceil,
\]
where
\[
R_{\BM}^{\mathrm{fid}}(n)=
\begin{cases}
1+\dfrac{1}{n-2}, & n \textnormal{ even},\\[8pt]
1+\dfrac{1}{n-1}, & n \textnormal{ odd},
\end{cases}
\]
then
\[
\Pr\left[\left|\widehat F_{\BM}^{\mathrm{fid}}-\overline F\right|\ge\eps\right]
\le\delta .
\]
In particular, for identical copies $\rho_1=\cdots=\rho_N=\rho$,
\[
\Pr\left[\left|\widehat F_{\BM}^{\mathrm{fid}}-F_{G_n}(\rho)\right|\ge\eps\right]
\le\delta .
\]
\end{corollary}

In Corollary~\ref{cor:bm-fid-copy-complexity}, the leading Hoeffding coefficient multiplying $\ln(2/\delta)/\eps^2$ is
\[
C_{\BM}^{\mathrm{fid}}(n)=
\begin{cases}
\dfrac12\left(1+\dfrac1{n-2}\right)^2, & n\ge4\text{ even},\\[12pt]
\dfrac12\left(1+\dfrac1{n-1}\right)^2, & n\ge3\text{ odd}.
\end{cases}
\]
Thus $C_{\BM}^{\mathrm{fid}}(n)\to1/2$ as $n\to\infty$, matching the leading Hoeffding coefficient of an ideal direct measurement of the projector $\proj{G_n}$.

\section{Certification protocol for a linear nearest neighbor architecture}
\label{sec:open-line-certification}

The near-projective behavior of BM-Cert comes from sampling matchings on the complete graph of qubits. In this section we analyze what remains possible when the hardware graph is instead the open path
\[
1-2-\cdots-n ,
\]
which is the standard linear nearest neighbor architecture \cite{saeedi2011synthesis, shafaei2013optimization, wille2014exact}.

On each copy, the verifier may partition \([n]\) into connected blocks of size one or two along the path, perform an arbitrary POVM on each block, and then accept or reject using arbitrary classical postprocessing of all outcomes. Thus every two-qubit block must be an adjacent pair \(\{j,j+1\}\), and the two-qubit blocks used in a single round form a matching of the path. As before, a randomized strategy is represented by its verification operator \(\Omega=\sum_s p_sE_s\), where each \(E_s\) is the pass effect of one such line-local disjoint 2-local test. More explicitly, if a setting \(s\) uses a block partition \(\mathcal P_s\), block POVMs \(\{F^{(b)}_{y_b}\}_{y_b}\) for \(b\in\mathcal P_s\), and an acceptance function \(f_s(y)\in[0,1]\), then
$
E_s=
\sum_y f_s(y)\bigotimes_{b\in\mathcal P_s}F^{(b)}_{y_b}.
$

For this open-boundary model define
\[
\beta_{\mathrm{line}}^*(n)
:=
\inf \bigl\{\beta(\Omega):\Omega \textnormal{ is achievable in the above model and }\Omega\ket{G_n}=\ket{G_n}\bigr\},
\]
and \(\nu_{\mathrm{line}}^*(n):=1-\beta_{\mathrm{line}}^*(n)\). We prove that, for every \(n\ge3\),
\[
\beta_{\mathrm{line}}^*(n)=\frac15,
\qquad
\nu_{\mathrm{line}}^*(n)=\frac45 .
\]
Thus open-boundary linear connectivity cannot reproduce the \(1-O(1/n)\) gap of all-to-all Bell matching, although it still allows a strictly larger gap than the local Pauli GHZ verification gap \(2/3\).

\begin{figure}
    \centering
    \includegraphics[width=0.69\linewidth]{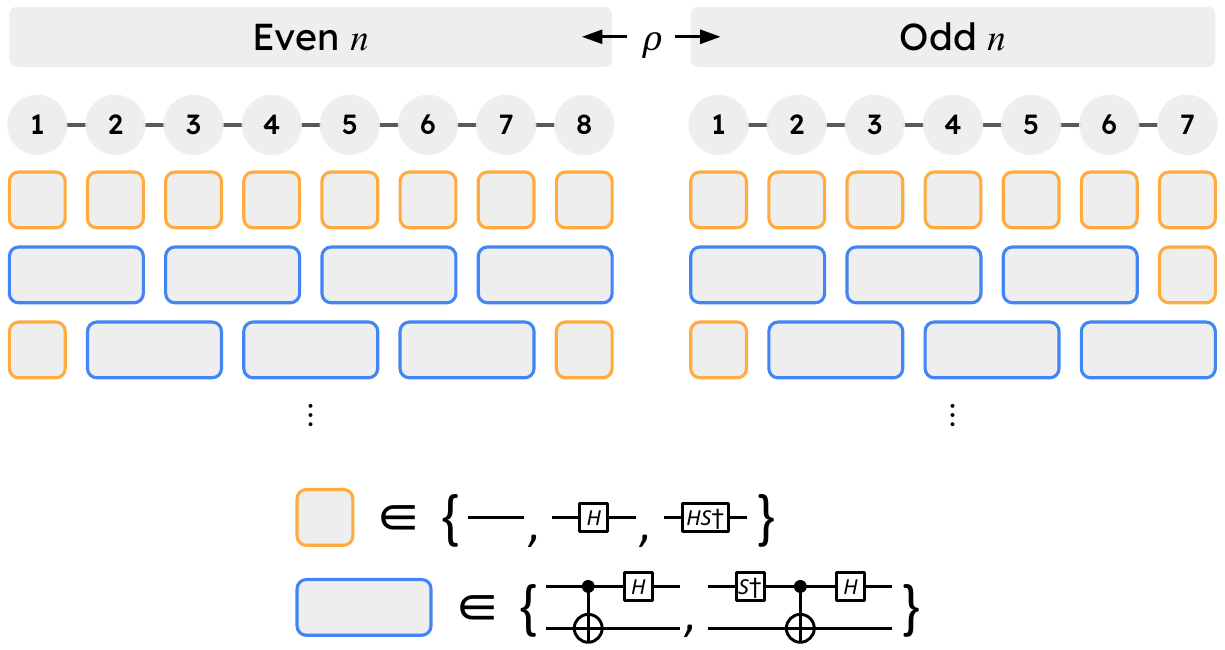}
    \caption{Measurement setups in the Brick-Cert protocol.}
    \label{fig:brick-cert-measurements}
\end{figure}

\subsection{A line-local protocol with spectral gap \(4/5\)}

Let
\begin{align}
\label{eq:odd-even-brick}
M_{\mathrm{odd}}
:=
\bigl\{\{1,2\},\{3,4\},\{5,6\},\ldots\bigr\},\nonumber\\
M_{\mathrm{even}}
:=
\bigl\{\{2,3\},\{4,5\},\{6,7\},\ldots\bigr\}
\end{align}
be the two brickwork matchings of the open path, with only pairs contained in \([n]\) included. For any path matching \(M\), define
\[
C_M:=\prod_{e\in M}P^Z_e,
\qquad
P^Z_{\{i,j\}}=\frac{\id+Z_iZ_j}{2},
\]
and define the two-dimensional population projector and the GHZ coherence operator by
\[
P_{\mathrm{eq}}:=\proj{0^n}+\proj{1^n},
\qquad
\Gamma_n:=\ket{0^n}\!\bra{1^n}+\ket{1^n}\!\bra{0^n}.
\]
Here \(P_{\mathrm{eq}}\) is implemented by measuring every qubit in the computational basis and accepting exactly the two outcomes \(0^n\) and \(1^n\).
We shall also use the Pauli matrix
$
Y=(\begin{smallmatrix}0&-i\\ i&0\end{smallmatrix}).
$

\begin{algorithm}[!htbp]
\caption{\justifying Brick-Cert for \(\ket{G_n}\), with \(n\ge3\)}
\label{alg:brick-cert}
\begin{algorithmic}[1]
\Require Infidelity threshold \(\eps\in(0,1)\) and significance level \(\delta\in(0,1)\)

\State Set
$
\nu_{\mathrm{brick}}=4/5
$
and
\[
N=\left\lceil\frac{\ln(1/\delta)}{-\ln(1-\nu_{\mathrm{brick}}\eps)}\right\rceil.
\]

\State Define the two brickwork matchings
$
M_{\mathrm{odd}}
$
and
$
M_{\mathrm{even}}
$
as in Eq.~\eqref{eq:odd-even-brick}.

\For{each round \(t=1,\ldots,N\)}
    \State Sample \(T_t\in\{\mathrm{pop},\mathrm{odd},\mathrm{even}\}\) with probabilities \(1/5,2/5,2/5\)

    \If{\(T_t=\mathrm{pop}\)}
        \State Measure every qubit in the computational basis
        \State Record the outcome \(a=(a_1,\ldots,a_n)\in\{0,1\}^n\)
        \State Accept this round if and only if \(a=0^n\) or \(a=1^n\)
    \Else
        \State Set \(M_t=M_{\mathrm{odd}}\) if \(T_t=\mathrm{odd}\), and \(M_t=M_{\mathrm{even}}\) if \(T_t=\mathrm{even}\)
        \State Let \(\mathcal B_t\) be the partition with pairs \(M_t\) and all unmatched singletons
        \State Sample a uniformly random even subset \(S_t\subseteq\mathcal B_t\)

        \For{every pair block \(b=\{i,i+1\}\in M_t\)}
            \If{\(b\in S_t\)}
                \State Jointly measure \(Z_iZ_{i+1}\) and \(Y_iX_{i+1}\)
            \Else
                \State Jointly measure \(Z_iZ_{i+1}\) and \(X_iX_{i+1}\)
            \EndIf
            \State Record outcomes \(z_b,h_b\in\{+1,-1\}\)
        \EndFor

        \For{every singleton block \(b=\{i\}\in\mathcal B_t\)}
            \If{\(b\in S_t\)}
                \State Measure \(Y_i\) and record its outcome \(h_b\in\{+1,-1\}\)
            \Else
                \State Measure \(X_i\) and record its outcome \(h_b\in\{+1,-1\}\)
            \EndIf
        \EndFor
\algstore{brickcert}
\end{algorithmic}
\end{algorithm}
\begin{algorithm}[!htbp]
\ContinuedFloat
\caption{\justifying Brick-Cert for \(\ket{G_n}\), with \(n\ge3\) \emph{(continued)}}
\begin{algorithmic}[1]
\algrestore{brickcert}
        \State Accept this round if and only if
        \[
        z_b=+1 \qquad \text{for all } b\in M_t
        \]
        \State and
        \[
        (-1)^{|S_t|/2}\prod_{b\in\mathcal B_t}h_b=+1.
        \]
    \EndIf
\EndFor

\Statex \textbf{Final decision:} Accept the preparation only if every one of the \(N\) rounds accepts
\end{algorithmic}
\end{algorithm}

\begin{lemma}[Line-local logical phase test]
\label{lem:line-logical-phase-test}
Let \(M\) be any matching of the open path. There is a line-local disjoint 2-local test with perfect completeness for \(\ket{G_n}\) whose pass effect is
\[
\Phi_M=\frac12 C_M+\frac12\Gamma_n .
\]
\end{lemma}

\begin{proof}
The matching \(M\) partitions the path into pair blocks \(e\in M\) and singleton blocks consisting of the unmatched vertices. Denote this block partition by \(\mathcal B(M)\), and let \(m:=|\mathcal B(M)|\). We define a logical qubit on each block.

For a singleton block \(b=\{i\}\), set
\[
\ket{0_b}:=\ket{0_i},
\qquad
\ket{1_b}:=\ket{1_i},
\qquad
X_b:=X_i,
\qquad
Y_b:=Y_i .
\]
For a pair block \(b=\{i,i+1\}\in M\), set
\[
\ket{0_b}:=\ket{0_i0_{i+1}},
\qquad
\ket{1_b}:=\ket{1_i1_{i+1}},
\qquad
X_b:=X_iX_{i+1},
\qquad
Y_b:=Y_iX_{i+1}.
\]
The code space of the pair block is the \(+1\) eigenspace of \(Z_iZ_{i+1}\), namely \(\operatorname{span}\{\ket{0_b},\ket{1_b}\}\). The two operators \(X_b\) and \(Y_b\) commute with \(Z_iZ_{i+1}\), since each contains an even number of single-qubit factors that anticommute with the corresponding \(Z\). On this two-dimensional code space they act as the usual logical Pauli matrices:
\[
X_b\ket{0_b}=\ket{1_b},
\quad
X_b\ket{1_b}=\ket{0_b},
\]
and
\[
Y_b\ket{0_b}=i\ket{1_b},
\quad
Y_b\ket{1_b}=-i\ket{0_b}.
\]
The same statements are immediate for singleton blocks with the ordinary single-qubit Pauli operators.

For each even subset \(S\subseteq\mathcal B(M)\), define the blockwise observable
\[
H_S:=(-1)^{|S|/2}
\prod_{b\in S}Y_b
\prod_{b\in\mathcal B(M)\setminus S}X_b .
\]
All factors act on disjoint blocks, and \(H_S\) commutes with every pair-code projector in \(C_M\). A valid line-local implementation is as follows. Sample an even subset \(S\subseteq\mathcal B(M)\) uniformly. On each pair block \(b=\{i,i+1\}\), jointly measure \(Z_iZ_{i+1}\) and either \(Y_b\), if \(b\in S\), or \(X_b\), if \(b\notin S\). On each singleton block, measure either \(Y_b\) or \(X_b\) according to the same rule. Accept if and only if all pair-code outcomes are \(+1\) and the measured eigenvalues have signed product \(+1\), with the sign \((-1)^{|S|/2}\). For this fixed \(S\), the pass effect is therefore
\[
C_M\frac{\id+H_S}{2}.
\]
Averaging over the \(2^{m-1}\) even subsets gives
\[
\Phi_M
=
\frac1{2^{m-1}}\sum_{S\subseteq\mathcal B(M):\ |S|\ \mathrm{even}}
C_M\frac{\id+H_S}{2}
=
\frac12 C_M+\frac12 C_M\overline H,
\]
where
\[
\overline H:=
\frac1{2^{m-1}}\sum_{S\subseteq\mathcal B(M):\ |S|\ \mathrm{even}}H_S.
\]
It remains to identify \(C_M\overline H\).

Choose an ordering of the blocks and write
\[
\ket{x}_L:=\bigotimes_{r=1}^m\ket{x_r}_{b_r},
\qquad
x=(x_1,\ldots,x_m)\in\{0,1\}^m,
\]
for the logical computational basis of the range of \(C_M\). On this code space, the block operators are ordinary logical Pauli operators. Since \(Y=iXZ\), for even \(S\) we have, on the code space,
\[
H_S=X_L^{\otimes m}\prod_{b_r\in S}Z_{L,r}.
\]
Hence
\[
\overline H
=
X_L^{\otimes m}
\left(
\frac1{2^{m-1}}\sum_{S:\ |S|\ \mathrm{even}}
\prod_{b_r\in S}Z_{L,r}
\right)
\]
on the code space. The operator in parentheses is diagonal in the logical computational basis. Its eigenvalue on \(\ket{x}_L\) is
\[
\frac1{2^{m-1}}\sum_{S:\ |S|\ \mathrm{even}}(-1)^{\sum_{b_r\in S}x_r}.
\]
This average is equal to \(1\) when \(x=0^m\) or \(x=1^m\), and equal to \(0\) otherwise. Therefore
\[
\frac1{2^{m-1}}\sum_{S:\ |S|\ \mathrm{even}}
\prod_{b_r\in S}Z_{L,r}
=
\proj{0_L^m}+\proj{1_L^m},
\]
and so
\[
\overline H
=
\ket{1_L^m}\!\bra{0_L^m}+\ket{0_L^m}\!\bra{1_L^m}
\]
on the code space. The logical all-zero and all-one vectors are the physical vectors \(\ket{0^n}\) and \(\ket{1^n}\). Consequently
\[
C_M\overline H
=
\ket{0^n}\!\bra{1^n}+\ket{1^n}\!\bra{0^n}
=
\Gamma_n .
\]
Substituting this into the averaged pass effect gives \(\Phi_M=\frac12C_M+\frac12\Gamma_n\). Since \(C_M\ket{G_n}=\ket{G_n}\) and \(\Gamma_n\ket{G_n}=\ket{G_n}\), the test has perfect completeness.
\end{proof}

Define the open-line verification operator
\[
\Omega_{\mathrm{line}}
:=
\frac15P_{\mathrm{eq}}
+
\frac25\Phi_{M_{\mathrm{odd}}}
+
\frac25\Phi_{M_{\mathrm{even}}}.
\]
Equivalently, using Lemma~\ref{lem:line-logical-phase-test},
\[
\Omega_{\mathrm{line}}
=
\frac15P_{\mathrm{eq}}
+
\frac15C_{M_{\mathrm{odd}}}
+
\frac15C_{M_{\mathrm{even}}}
+
\frac25\Gamma_n .
\]
Operationally, the verifier performs the population test with probability \(1/5\), the logical phase test for \(M_{\mathrm{odd}}\) with probability \(2/5\), and the logical phase test for \(M_{\mathrm{even}}\) with probability \(2/5\).
We call this open-boundary line-local certification protocol Brick-Cert. The measurements for both even and odd $n$ are visualized in Figure~\ref{fig:brick-cert-measurements}, and the full verification protocol is detailed in Algorithm~\ref{alg:brick-cert}.

\begin{proposition}[Spectrum of the open-line protocol]
\label{prop:line-protocol-spectrum}
For every \(n\ge3\), the verification operator \(\Omega_{\mathrm{line}}\) has perfect completeness for \(\ket{G_n}\). Its second eigenvalue and spectral gap are
\[
\beta(\Omega_{\mathrm{line}})=\frac15,
\qquad
\nu(\Omega_{\mathrm{line}})=\frac45 .
\]
\end{proposition}

\begin{proof}
Recall the GHZ-basis vectors
$
\ket{a,\tau}
$
in Eq.~\eqref{eq:ghz-basis-def}.
These vectors form an orthonormal basis.
Also recall Eq.~\eqref{eq:apply-zz-to-ghz-basis}.
For a bit string \(a\), define its open-path boundary by
\[
\partial a:=\{j\in\{1,\ldots,n-1\}:a_j\ne a_{j+1}\}.
\]
If \(M\) is a path matching, then
\[
C_M\ket{a,\tau}
=
\begin{cases}
\ket{a,\tau}, & \partial a\cap\{j:\{j,j+1\}\in M\}=\varnothing,\\[3pt]
0, & \textnormal{otherwise}.
\end{cases}
\]
Also,
\[
P_{\mathrm{eq}}\ket{0^n,\tau}=\ket{0^n,\tau},
\qquad
\Gamma_n\ket{0^n,\tau}=\tau\ket{0^n,\tau},
\]
while \(P_{\mathrm{eq}}\ket{a,\tau}=0\) and \(\Gamma_n\ket{a,\tau}=0\) whenever \(a\) is nonconstant.

The target is \(\ket{G_n}=\ket{0^n,+}\). It is fixed by \(P_{\mathrm{eq}}\), both code projectors \(C_{M_{\mathrm{odd}}}\) and \(C_{M_{\mathrm{even}}}\), and \(\Gamma_n\). Therefore
\[
\Omega_{\mathrm{line}}\ket{G_n}=\ket{G_n}.
\]
For the phase-flipped GHZ vector \(\ket{0^n,-}\), the first three projectors still have eigenvalue \(1\), whereas \(\Gamma_n\) has eigenvalue \(-1\). Hence
\[
\Omega_{\mathrm{line}}\ket{0^n,-}
=
\left(\frac15+\frac15+\frac15-\frac25\right)\ket{0^n,-}
=
\frac15\ket{0^n,-}.
\]

Now let \(a\) be nonconstant. Then \(\partial a\ne\varnothing\), and the two brickwork matchings partition the path edges. If \(\partial a\) contains both an odd edge and an even edge, then both \(C_{M_{\mathrm{odd}}}\) and \(C_{M_{\mathrm{even}}}\) reject \(\ket{a,\tau}\), so the eigenvalue is \(0\). If \(\partial a\) is nonempty and contained entirely in one of the two brickwork layers, exactly one of the two code projectors accepts and the other rejects, so the eigenvalue is \(1/5\). Thus every GHZ-basis vector orthogonal to \(\ket{G_n}\) has eigenvalue either \(0\) or \(1/5\).
Therefore the largest eigenvalue on the subspace orthogonal to \(\ket{G_n}\) is \(1/5\), proving the claim.
\end{proof}

\begin{corollary}[Copy complexity of Brick-Cert]
\label{cor:sample-complexity-line}
Let \(\eps,\delta\in(0,1)\). Suppose Brick-Cert is run for
\[
N_{\mathrm{brick}}(\eps,\delta)=
\left\lceil
\frac{\ln(1/\delta)}{-\ln(1-\nu(\Omega_{\mathrm{line}})\eps)}
\right\rceil
\]
independent rounds, accepting only if all rounds pass. If the average infidelity of the tested copies with \(\ket{G_n}\) is at least \(\eps\), then the probability of accepting is at most \(\delta\). If every tested copy is exactly \(\ket{G_n}\), then the probability of accepting is \(1\).
\end{corollary}

\subsection{Optimality of Brick-Cert}

We now prove that no line-local disjoint 2-local certification strategy can have spectral gap larger than \(4/5\). Define
\[
\ket{G_n^-}:=\frac{\ket{0^n}-\ket{1^n}}{\sqrt2}.
\]
For each cut edge \(j\in\{1,\ldots,n-1\}\), define the two domain-wall GHZ states
\[
\ket{D_j^+}:=\frac{\ket{0^j1^{n-j}}+\ket{1^j0^{n-j}}}{\sqrt2},
\qquad
\ket{D_j^-}:=\frac{\ket{0^j1^{n-j}}-\ket{1^j0^{n-j}}}{\sqrt2}.
\]
The four states \(\ket{G_n}\), \(\ket{G_n^-}\), \(\ket{D_j^+}\), and \(\ket{D_j^-}\) are the Bell basis of the two-qubit logical subspace across the bipartition \(\{1,\ldots,j\}\mid\{j+1,\ldots,n\}\) generated by the all-zero and all-one strings on each side.

\begin{lemma}
\label{lem:line-missed-cut-inequality}
Fix \(j\in\{1,\ldots,n-1\}\). Let \(E\) be a positive semidefinite operator satisfying \(0\le E\le\id\) and \(E\ket{G_n}=\ket{G_n}\). If \(E\) is separable across the bipartition
\[
\{1,\ldots,j\}\mid\{j+1,\ldots,n\},
\]
then
\[
\bra{G_n^-}E\ket{G_n^-}
+
\bra{D_j^+}E\ket{D_j^+}
+
\bra{D_j^-}E\ket{D_j^-}
\ge1 .
\]
\end{lemma}

\begin{proof}
Write \(L=\{1,\ldots,j\}\) and \(R=\{j+1,\ldots,n\}\). Since \(E\) is separable across \(L\mid R\), it can be written as a finite conic combination of rank-one product operators,
\[
E=\sum_\ell c_\ell
\proj{\alpha_\ell}_L\otimes\proj{\beta_\ell}_R,
\qquad
c_\ell\ge0,
\]
where the vectors \(\ket{\alpha_\ell}\) and \(\ket{\beta_\ell}\) need not be normalized. It is therefore enough to prove the corresponding inequality term by term for an arbitrary product vector \(\ket\alpha_L\otimes\ket\beta_R\).

Let
\[
\alpha_0:=\braket{0^j}{\alpha},
\quad
\alpha_1:=\braket{1^j}{\alpha},
\quad
\beta_0:=\braket{0^{n-j}}{\beta},
\quad
\beta_1:=\braket{1^{n-j}}{\beta}.
\]
The four states \(\ket{G_n}\), \(\ket{G_n^-}\), \(\ket{D_j^+}\), and \(\ket{D_j^-}\) form an orthonormal basis of
\[
\operatorname{span}\bigl\{
\ket{0^j0^{n-j}},
\ket{1^j1^{n-j}},
\ket{0^j1^{n-j}},
\ket{1^j0^{n-j}}
\bigr\}.
\]
Hence the total squared overlap of \(\ket\alpha_L\ket\beta_R\) with these four states is
\[
(|\alpha_0|^2+|\alpha_1|^2)(|\beta_0|^2+|\beta_1|^2).
\]
On the other hand,
\[
2\bigl|\bra{G_n}(\ket\alpha_L\ket\beta_R)\bigr|^2
=
\bigl|\alpha_0\beta_0+\alpha_1\beta_1\bigr|^2
\le
(|\alpha_0|^2+|\alpha_1|^2)(|\beta_0|^2+|\beta_1|^2)
\]
by the Cauchy--Schwarz inequality. Therefore, for every rank-one product operator \(F=\proj\alpha_L\otimes\proj\beta_R\),
\[
2\bra{G_n}F\ket{G_n}
\le
\bra{G_n}F\ket{G_n}
+
\bra{G_n^-}F\ket{G_n^-}
+
\bra{D_j^+}F\ket{D_j^+}
+
\bra{D_j^-}F\ket{D_j^-}.
\]
Multiplying by \(c_\ell\) and summing over \(\ell\) gives
\[
2\bra{G_n}E\ket{G_n}
\le
\bra{G_n}E\ket{G_n}
+
\bra{G_n^-}E\ket{G_n^-}
+
\bra{D_j^+}E\ket{D_j^+}
+
\bra{D_j^-}E\ket{D_j^-}.
\]
Finally, \(E\ket{G_n}=\ket{G_n}\) implies \(\bra{G_n}E\ket{G_n}=1\). Cancelling one copy of this term gives the desired inequality.
\end{proof}

\begin{theorem}
\label{thm:optimal-open-line-certification}
For every \(n\ge3\), the optimal second eigenvalue and spectral gap among all open-boundary line-local disjoint 2-local certification strategies with perfect completeness for \(\ket{G_n}\) are
\[
\beta_{\mathrm{line}}^*(n)=\frac15,
\qquad
\nu_{\mathrm{line}}^*(n)=\frac45 .
\]
\end{theorem}

\begin{proof}
The upper bound \(\beta_{\mathrm{line}}^*(n)\le1/5\) follows from Proposition~\ref{prop:line-protocol-spectrum}. It remains to prove the converse.

Let \(\Omega=\sum_s p_sE_s\) be any open-boundary line-local disjoint 2-local verification operator with \(\Omega\ket{G_n}=\ket{G_n}\). By Lemma~\ref{lem:perfect-completeness-round-by-round}, every setting \(E_s\) used with positive probability satisfies \(E_s\ket{G_n}=\ket{G_n}\). For each setting \(s\), let \(M_s\subseteq\{1,\ldots,n-1\}\) be the set of path edges \(j\) such that the two-qubit block \(\{j,j+1\}\) is measured in that setting. Because the two-qubit blocks are disjoint, \(M_s\) represents a matching of the path.

For a cut edge \(j\), if \(j\notin M_s\), then no measured block crosses the bipartition \(\{1,\ldots,j\}\mid\{j+1,\ldots,n\}\). The pass effect \(E_s\) is then separable across this bipartition. Therefore Lemma~\ref{lem:line-missed-cut-inequality} gives, for every \(j\notin M_s\),
\[
\bra{D_j^+}E_s\ket{D_j^+}
+
\bra{D_j^-}E_s\ket{D_j^-}
\ge
1-u_s,
\]
where
\[
u_s:=\bra{G_n^-}E_s\ket{G_n^-}.
\]
Let
\[
U:=\bra{G_n^-}\Omega\ket{G_n^-}=\sum_s p_su_s.
\]
For each path edge \(j\), set
\[
V_{j,+}:=\bra{D_j^+}\Omega\ket{D_j^+},
\qquad
V_{j,-}:=\bra{D_j^-}\Omega\ket{D_j^-}.
\]
The preceding inequality implies
\[
V_{j,+}+V_{j,-}
\ge
\sum_{s:\ j\notin M_s}p_s(1-u_s).
\]
Now define nonnegative weights
\[
w_s:=p_s(1-u_s),
\qquad
W:=\sum_s w_s=1-U,
\]
and, for each edge \(j\),
\[
r_j:=\sum_{s:\ j\in M_s}w_s.
\]
Since every \(M_s\) represents a matching, adjacent path edges cannot both belong to \(M_s\). Thus for every \(j=1,\ldots,n-2\),
\[
r_j+r_{j+1}\le W .
\]
Because \(n\ge3\), the path has at least two adjacent edges, and hence at least one edge \(j_*\) satisfies \(r_{j_*}\le W/2\). For this edge,
\[
V_{j_*,+}+V_{j_*,-}
\ge
\sum_{s:\ j_*\notin M_s}w_s
=
W-r_{j_*}
\ge
\frac W2
=
\frac{1-U}{2}.
\]
Consequently at least one sign \(\sigma\in\{+,-\}\) satisfies
\[
V_{j_*,\sigma}\ge\frac{1-U}{4}.
\]
The states \(\ket{G_n^-}\) and \(\ket{D_{j_*}^{\sigma}}\) are both orthogonal to \(\ket{G_n}\). Since \(\Omega\) is positive semidefinite,
\[
\beta(\Omega)
\ge
\max\left\{
\bra{G_n^-}\Omega\ket{G_n^-},
\bra{D_{j_*}^{\sigma}}\Omega\ket{D_{j_*}^{\sigma}}
\right\}
\ge
\max\left\{U,\frac{1-U}{4}\right\}.
\]
If \(U\ge1/5\), this maximum is at least \(1/5\). If \(U\le1/5\), then \((1-U)/4\ge1/5\). Hence \(\beta(\Omega)\ge1/5\) for every admissible \(\Omega\), proving \(\beta_{\mathrm{line}}^*(n)\ge1/5\). Together with the construction, this gives \(\beta_{\mathrm{line}}^*(n)=1/5\), and therefore \(\nu_{\mathrm{line}}^*(n)=4/5\).
\end{proof}

\section{Experiments}

\begin{figure}[t]
    \centering

    \begin{subfigure}[b]{0.15\textwidth}
        \centering
        \includegraphics[width=\textwidth]{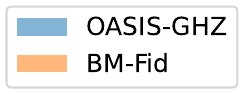}
        \label{fig:top_center}
    \end{subfigure}
    

    \begin{subfigure}[b]{0.32\textwidth}
        \centering
        \includegraphics[width=\textwidth]{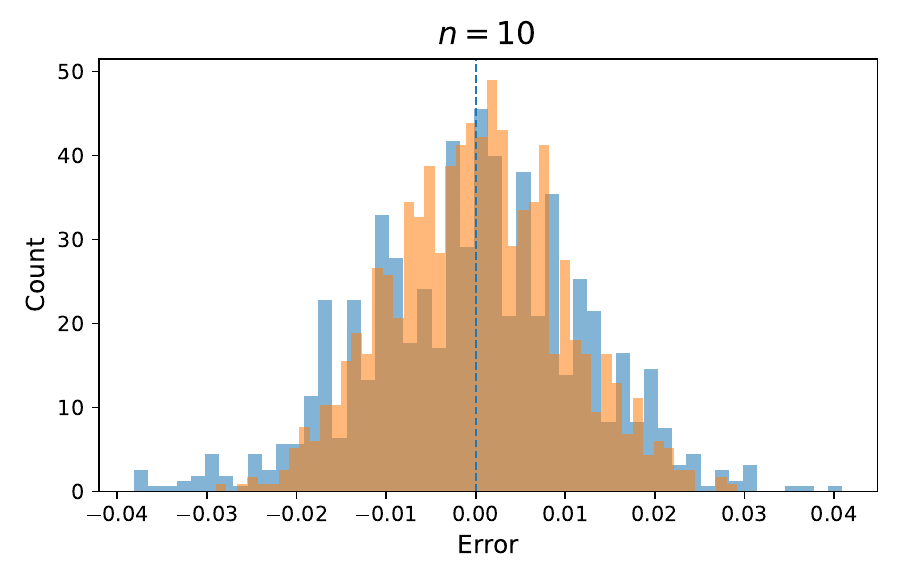}
        \label{fig:sub1}
    \end{subfigure}
    \hfill
    \begin{subfigure}[b]{0.32\textwidth}
        \centering
        \includegraphics[width=\textwidth]{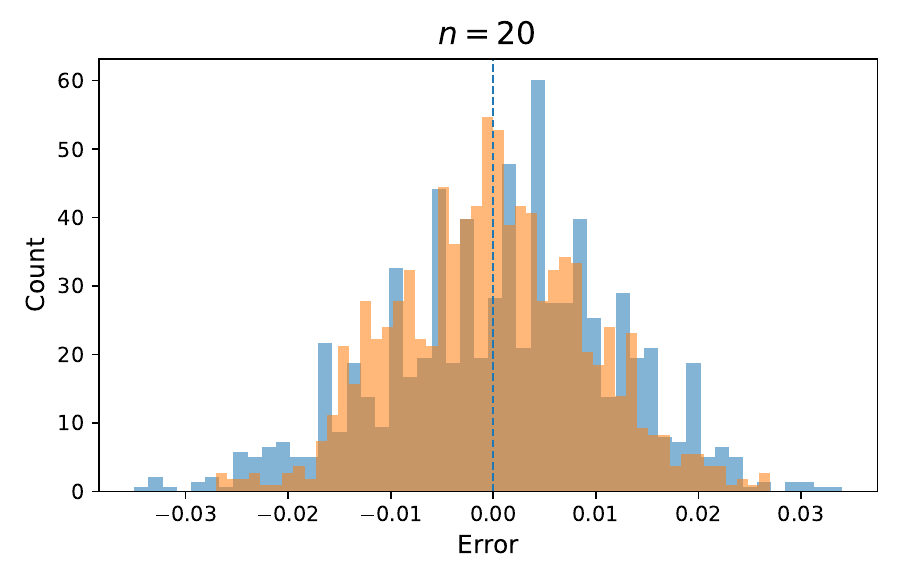}
        \label{fig:sub2}
    \end{subfigure}
    \hfill
    \begin{subfigure}[b]{0.32\textwidth}
        \centering
        \includegraphics[width=\textwidth]{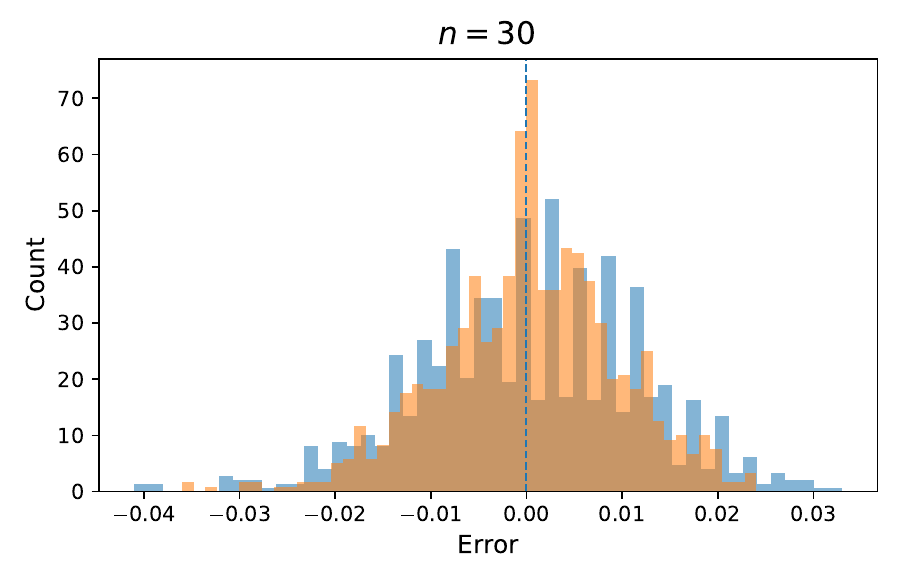}
        \label{fig:sub3}
    \end{subfigure}
    
    \vspace{0.3cm}
    
    \begin{subfigure}[b]{0.32\textwidth}
        \centering
        \includegraphics[width=\textwidth]{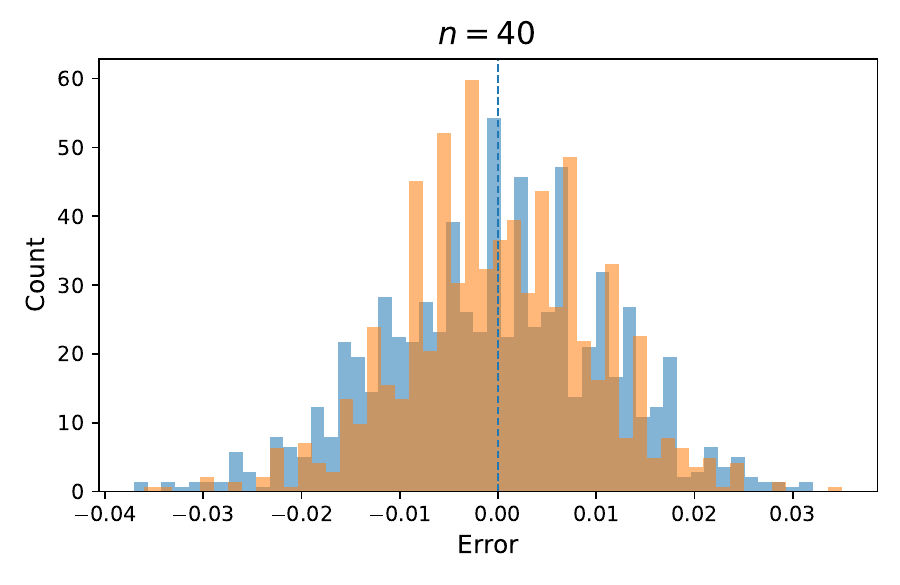}
        \label{fig:sub4}
    \end{subfigure}
    \hfill
    \begin{subfigure}[b]{0.32\textwidth}
        \centering
        \includegraphics[width=\textwidth]{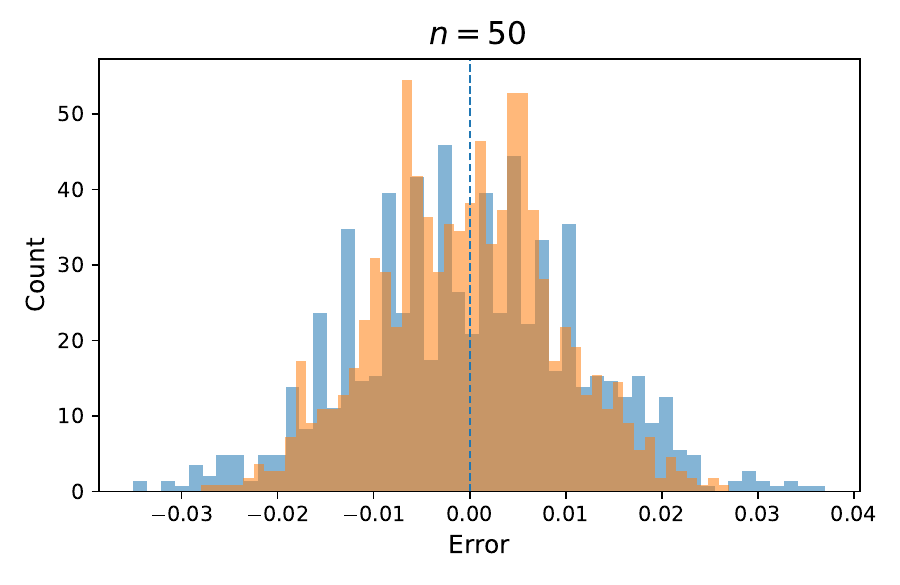}
        \label{fig:sub5}
    \end{subfigure}
    \hfill
    \begin{subfigure}[b]{0.32\textwidth}
        \centering
        \includegraphics[width=\textwidth]{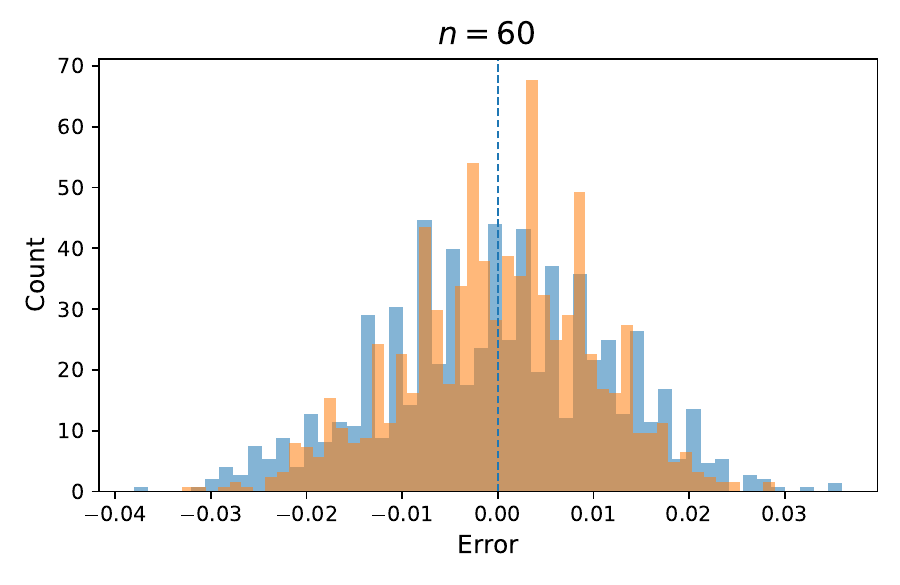}
        \label{fig:sub6}
    \end{subfigure}

    \caption{\justifying Histograms of errors in GHZ fidelity estimation for OASIS-GHZ \cite{cha2026operator} and BM-Fid under a matched copy budget.}
    \label{fig:main_grid}
\end{figure}

We compare BM-Fid against the OASIS-GHZ fidelity estimator in \cite{cha2026operator}.
We focus on fidelity estimation because it has a direct quantitative performance metric, namely the mean squared error of the estimator.
For each
$
n\in\{10,\allowbreak20,\allowbreak30,\allowbreak40,\allowbreak50,\allowbreak60\},
$
we used the depolarized GHZ state
$$
\rho_n
=0.9\,\proj{G_n}
+0.1\,\frac{\id}{2^n}.
$$
For each value of \(n\), each method was given the same copy budget, with one independent estimation experiment using \(N=1000\) copies of \(\rho_n\).  We repeated this independent \(N\)-copy experiment 1000 times for OASIS-GHZ and 1000 times for BM-Fid.  For every independent run, we recorded the estimation error
$
\widehat F - F_{G_n}(\rho_n),
$
where \(\widehat F\) denotes the fidelity estimate produced by the corresponding method.

Figure~\ref{fig:main_grid} shows the resulting error distributions.  Across all tested system sizes, BM-Fid gives a narrower error distribution than OASIS-GHZ under the same copy budget. Quantitatively, BM-Fid reduced the empirical MSE by more than 30\% for every tested value of \(n\).

\section{Conclusion}

In this work, we introduced BM-Cert, a single-copy verification protocol for the \(n\)-qubit GHZ state using only disjoint two-qubit Bell-basis measurements, together with one single-qubit \(X\)-basis measurement when \(n\) is odd.
The remarkable point is that once we allow only the additional capability of two-qubit entangling measurements, the spectral gap improves to \(1-O(1/n)\). Consequently, the leading copy coefficient approaches that of the ideal projective certification asymptotically.
BM-Fid further shows that the Bell-matching measurement primitive is not exhausted by its pass/fail verifier. With different classical postprocessing, it also performs fidelity estimation at almost the same cost as an ideal projection measurement.
Finally, Brick-Cert shows that under open-boundary linear nearest neighbor connectivity, disjoint 2-local certification achieves the optimal spectral gap \(4/5\).

Several directions remain open.
The BM-Cert analysis assumes uniformly random matchings on the complete interaction graph, which requires the ability to Bell-measure arbitrary pairs. Brick-Cert complements this by treating the open-boundary linear nearest neighbor architecture. More general connectivity graphs, however, would require a different analysis.
Another limitation is that memory-assisted and collective-copy strategies are outside the present model \cite{o2026instance}. Such strategies can be stronger, as demonstrated by recent work on quantum-memory-assisted verification \cite{chen2025quantum}. Finally, while we prove optimality inside the Bell-matching family, it remains open whether BM-Cert is optimal among all perfect-completeness verification strategies based on at most 2-local measurements.

\section*{Acknowledgments}

This work is in part supported by the National Research Foundation of Korea (NRF, RS-2024-00451435 (20\%), RS-2024-00413957 (20\%)), Institute of Information \& communications Technology Planning \& Evaluation (IITP, RS-2025-02305453 (15\%), RS-2025-02273157 (15\%), RS-2025-25442149 (15\%), RS-2021-II211343 (15\%)) grant funded by the Ministry of Science and ICT (MSIT), Institute of New Media and Communications (INMAC), and the BK21 FOUR program of the Education, Artificial Intelligence Graduate School Program (Seoul National University), and Research Program for Future ICT Pioneers, Seoul National University in 2026.

\section*{Author contributions}

H.C. conceptualized the framework, conducted the experiments, and wrote the manuscript. J.L. provided supervision and reviewed the manuscript.

\section*{Data availability}
The data that support the findings of this article are publicly available \cite{cha2026bmfidcode}.

\bibliography{main}

@incollection{greenberger1989going,
  title={{Going beyond Bell's theorem}},
  author={Greenberger, Daniel M and Horne, Michael A and Zeilinger, Anton},
  booktitle={Bell's theorem, quantum theory and conceptions of the universe},
  pages={69--72},
  year={1989},
  publisher={Springer}
}

@article{flammia2011direct,
  title={{Direct fidelity estimation from few Pauli measurements}},
  author={Flammia, Steven T and Liu, Yi-Kai},
  journal={Physical Review Letters},
  volume={106},
  number={23},
  pages={230501},
  year={2011},
  publisher={APS}
}

@article{pallister2018optimal,
  title={{Optimal verification of entangled states with local measurements}},
  author={Pallister, Sam and Linden, Noah and Montanaro, Ashley},
  journal={Physical Review Letters},
  volume={120},
  number={17},
  pages={170502},
  year={2018},
  publisher={APS}
}

@article{zhu2019efficient,
  title = {{Efficient Verification of Pure Quantum States in the Adversarial Scenario}},
  author = {Zhu, Huangjun and Hayashi, Masahito},
  journal = {Physical Review Letters},
  volume = {123},
  issue = {26},
  pages = {260504},
  year = {2019},
  publisher = {APS}
}

@article{li2020optimal,
  title={{Optimal verification of Greenberger--Horne--Zeilinger states}},
  author={Li, Zihao and Han, Yun-Guang and Zhu, Huangjun},
  journal={Physical Review Applied},
  volume={13},
  number={5},
  pages={054002},
  year={2020},
  publisher={APS}
}

@article{dangniam2020optimal,
  title={{Optimal verification of stabilizer states}},
  author={Dangniam, Ninnat and Han, Yun-Guang and Zhu, Huangjun},
  journal={Physical Review Research},
  volume={2},
  number={4},
  pages={043323},
  year={2020},
  publisher={APS}
}

@article{huang2020predicting,
  title={{Predicting many properties of a quantum system from very few measurements}},
  author={Huang, Hsin-Yuan and Kueng, Richard and Preskill, John},
  journal={Nature Physics},
  volume={16},
  number={10},
  pages={1050--1057},
  year={2020},
  publisher={Nature Publishing Group UK London}
}

@article{yu2022statistical,
  title={{Statistical methods for quantum state verification and fidelity estimation}},
  author={Yu, Xiao-Dong and Shang, Jiangwei and G{\"u}hne, Otfried},
  journal={Advanced Quantum Technologies},
  volume={5},
  number={5},
  pages={2100126},
  year={2022},
  publisher={Wiley Online Library}
}

@article{han2021optimal,
  title={{Optimal verification of the Bell state and Greenberger--Horne--Zeilinger states in untrusted quantum networks}},
  author={Han, Yun-Guang and Li, Zihao and Wang, Yukun and Zhu, Huangjun},
  journal={npj Quantum Information},
  volume={7},
  number={1},
  pages={164},
  year={2021},
  publisher={Nature Publishing Group UK London}
}

@article{liu2021universally,
  title={{Universally optimal verification of entangled states with nondemolition measurements}},
  author={Liu, Ye-Chao and Shang, Jiangwei and Han, Rui and Zhang, Xiangdong},
  journal={Physical Review Letters},
  volume={126},
  number={9},
  pages={090504},
  year={2021},
  publisher={APS}
}

@article{chen2025quantum,
  title={{Quantum memory assisted entangled state verification with local measurements}},
  author={Chen, Siyuan and Xie, Wei and Xu, Ping and Wang, Kun},
  journal={Physical Review Research},
  volume={7},
  number={1},
  pages={013003},
  year={2025},
  publisher={APS}
}

@article{aolita2015reliable,
  title={{Reliable quantum certification of photonic state preparations}},
  author={Aolita, Leandro and Gogolin, Christian and Kliesch, Martin and Eisert, Jens},
  journal={Nature Communications},
  volume={6},
  number={1},
  pages={8498},
  year={2015},
  publisher={Nature Publishing Group UK London}
}

@article{kliesch2021theory,
  title={{Theory of quantum system certification}},
  author={Kliesch, Martin and Roth, Ingo},
  journal={PRX Quantum},
  volume={2},
  number={1},
  pages={010201},
  year={2021},
  publisher={APS}
}

@article{greenberger1990bell,
  title={{Bell's theorem without inequalities}},
  author={Greenberger, Daniel M and Horne, Michael A and Shimony, Abner and Zeilinger, Anton},
  journal={American Journal of Physics},
  volume={58},
  number={12},
  pages={1131--1143},
  year={1990},
  publisher={American Association of Physics Teachers}
}

@article{mermin1990extreme,
  title={{Extreme quantum entanglement in a superposition of macroscopically distinct states}},
  author={Mermin, N David},
  journal={Physical Review Letters},
  volume={65},
  number={15},
  pages={1838},
  year={1990},
  publisher={APS}
}

@article{bouwmeester1999observation,
  title={{Observation of three-photon Greenberger-Horne-Zeilinger entanglement}},
  author={Bouwmeester, Dik and Pan, Jian-Wei and Daniell, Matthew and Weinfurter, Harald and Zeilinger, Anton},
  journal={Physical Review Letters},
  volume={82},
  number={7},
  pages={1345},
  year={1999},
  publisher={APS}
}

@article{pan2000experimental,
  title={{Experimental test of quantum nonlocality in three-photon Greenberger--Horne--Zeilinger entanglement}},
  author={Pan, Jian-Wei and Bouwmeester, Dik and Daniell, Matthew and Weinfurter, Harald and Zeilinger, Anton},
  journal={Nature},
  volume={403},
  number={6769},
  pages={515--519},
  year={2000},
  publisher={Nature Publishing Group UK London}
}

@article{hillery1999quantum,
  title={{Quantum secret sharing}},
  author={Hillery, Mark and Bu{\v{z}}ek, Vladim{\'\i}r and Berthiaume, Andr{\'e}},
  journal={Physical Review A},
  volume={59},
  number={3},
  pages={1829},
  year={1999},
  publisher={APS}
}

@article{karlsson1998quantum,
  title={{Quantum teleportation using three-particle entanglement}},
  author={Karlsson, Anders and Bourennane, Mohamed},
  journal={Physical Review A},
  volume={58},
  number={6},
  pages={4394},
  year={1998},
  publisher={APS}
}

@article{pickston2023conference,
  title={{Conference key agreement in a quantum network}},
  author={Pickston, Alexander and Ho, Joseph and Ulibarrena, Andr{\'e}s and Grasselli, Federico and Proietti, Massimiliano and Morrison, Christopher L and Barrow, Peter and Graffitti, Francesco and Fedrizzi, Alessandro},
  journal={npj Quantum Information},
  volume={9},
  number={1},
  pages={82},
  year={2023},
  publisher={Nature Publishing Group UK London}
}

@article{toth2005entanglement,
  title={{Entanglement detection in the stabilizer formalism}},
  author={T{\'o}th, G{\'e}za and G{\"u}hne, Otfried},
  journal={Physical Review A},
  volume={72},
  number={2},
  pages={022340},
  year={2005},
  publisher={APS}
}

@inproceedings{haah2016sample,
  title={{Sample-optimal tomography of quantum states}},
  author={Haah, Jeongwan and Harrow, Aram W and Ji, Zhengfeng and Wu, Xiaodi and Yu, Nengkun},
  booktitle={Proceedings of the Forty-Eighth Annual ACM Symposium on Theory of Computing},
  pages={913--925},
  year={2016}
}

@inproceedings{o2016efficient,
  title={{Efficient quantum tomography}},
  author={O'Donnell, Ryan and Wright, John},
  booktitle={Proceedings of the Forty-Eighth Annual ACM Symposium on Theory of Computing},
  pages={899--912},
  year={2016}
}

@article{cramer2010efficient,
  title={{Efficient quantum state tomography}},
  author={Cramer, Marcus and Plenio, Martin B and Flammia, Steven T and Somma, Rolando and Gross, David and Bartlett, Stephen D and Landon-Cardinal, Olivier and Poulin, David and Liu, Yi-Kai},
  journal={Nature Communications},
  volume={1},
  number={1},
  pages={149},
  year={2010},
  publisher={Nature Publishing Group UK London}
}

@article{gross2010quantum,
  title={{Quantum state tomography via compressed sensing}},
  author={Gross, David and Liu, Yi-Kai and Flammia, Steven T and Becker, Stephen and Eisert, Jens},
  journal={Physical Review Letters},
  volume={105},
  number={15},
  pages={150401},
  year={2010},
  publisher={APS}
}

@article{da2011practical,
  title={{Practical characterization of quantum devices without tomography}},
  author={da Silva, Marcus P and Landon-Cardinal, Olivier and Poulin, David},
  journal={Physical Review Letters},
  volume={107},
  number={21},
  pages={210404},
  year={2011},
  publisher={APS}
}

@inproceedings{aaronson2018shadow,
  title={{Shadow tomography of quantum states}},
  author={Aaronson, Scott},
  booktitle={Proceedings of the 50th Annual ACM SIGACT Symposium on Theory of Computing},
  pages={325--338},
  year={2018}
}

@article{bennett1993teleporting,
  title={{Teleporting an unknown quantum state via dual classical and Einstein-Podolsky-Rosen channels}},
  author={Bennett, Charles H and Brassard, Gilles and Cr{\'e}peau, Claude and Jozsa, Richard and Peres, Asher and Wootters, William K},
  journal={Physical Review Letters},
  volume={70},
  number={13},
  pages={1895},
  year={1993},
  publisher={APS}
}

@article{cha2025efficient,
  title={{Efficient sampling for Pauli-measurement-based shadow tomography in direct fidelity estimation}},
  author={Cha, Hyunho and Lee, Jungwoo},
  journal={Physical Review A},
  volume={112},
  number={3},
  pages={032427},
  year={2025},
  publisher={APS}
}

@article{cha2026operator,
  title={{Operator-aware shadow importance sampling for accurate fidelity estimation}},
  author={Cha, Hyunho and Hong, Sangwoo and Lee, Jungwoo},
  journal={Physical Review A},
  volume={113},
  number={4},
  pages={042602},
  year={2026},
  publisher={APS}
}

@article{barbera2025sampling,
  title={{Sampling groups of pauli operators to enhance direct fidelity estimation}},
  author={Barber{\`a}-Rodr{\'\i}guez, J{\'u}lia and Navarro, Mariana and Zambrano, Leonardo},
  journal={Quantum},
  volume={9},
  pages={1784},
  year={2025},
  publisher={Verein zur F{\"o}rderung des Open Access Publizierens in den Quantenwissenschaften}
}

@article{takeuchi2018verification,
  title={{Verification of many-qubit states}},
  author={Takeuchi, Yuki and Morimae, Tomoyuki},
  journal={Physical Review X},
  volume={8},
  number={2},
  pages={021060},
  year={2018},
  publisher={APS}
}

@article{govcanin2022sample,
  title={{Sample-efficient device-independent quantum state verification and certification}},
  author={Go{\v{c}}anin, Aleksandra and {\v{S}}upi{\'c}, Ivan and Daki{\'c}, Borivoje},
  journal={PRX Quantum},
  volume={3},
  number={1},
  pages={010317},
  year={2022},
  publisher={APS}
}

@article{nguyen2022optimizing,
  title={{Optimizing shadow tomography with generalized measurements}},
  author={Nguyen, H Chau and B{\"o}nsel, Jan Lennart and Steinberg, Jonathan and G{\"u}hne, Otfried},
  journal={Physical Review Letters},
  volume={129},
  number={22},
  pages={220502},
  year={2022},
  publisher={APS}
}

@article{hu2023classical,
  title={{Classical shadow tomography with locally scrambled quantum dynamics}},
  author={Hu, Hong-Ye and Choi, Soonwon and You, Yi-Zhuang},
  journal={Physical Review Research},
  volume={5},
  number={2},
  pages={023027},
  year={2023},
  publisher={APS}
}

@article{king2025triply,
  title = {{Triply Efficient Shadow Tomography}},
  author = {King, Robbie and Gosset, David and Kothari, Robin and Babbush, Ryan},
  journal = {PRX Quantum},
  volume = {6},
  issue = {1},
  pages = {010336},
  year = {2025},
  publisher = {APS}
}

@article{eisert2020quantum,
  title={{Quantum certification and benchmarking}},
  author={Eisert, Jens and Hangleiter, Dominik and Walk, Nathan and Roth, Ingo and Markham, Damian and Parekh, Rhea and Chabaud, Ulysse and Kashefi, Elham},
  journal={Nature Reviews Physics},
  volume={2},
  number={7},
  pages={382--390},
  year={2020},
  publisher={Nature Publishing Group UK London}
}

@article{james2001measurement,
  title={{Measurement of qubits}},
  author={James, Daniel FV and Kwiat, Paul G and Munro, William J and White, Andrew G},
  journal={Physical Review A},
  volume={64},
  number={5},
  pages={052312},
  year={2001},
  publisher={APS}
}

@article{christandl2012reliable,
  title={{Reliable quantum state tomography}},
  author={Christandl, Matthias and Renner, Renato},
  journal={Physical Review Letters},
  volume={109},
  number={12},
  pages={120403},
  year={2012},
  publisher={APS}
}

@article{blume2010optimal,
  title={{Optimal, reliable estimation of quantum states}},
  author={Blume-Kohout, Robin},
  journal={New Journal of Physics},
  volume={12},
  number={4},
  pages={043034},
  year={2010}
}

@article{zhang2021direct,
  title={{Direct fidelity estimation of quantum states using machine learning}},
  author={Zhang, Xiaoqian and Luo, Maolin and Wen, Zhaodi and Feng, Qin and Pang, Shengshi and Luo, Weiqi and Zhou, Xiaoqi},
  journal={Physical Review Letters},
  volume={127},
  number={13},
  pages={130503},
  year={2021},
  publisher={APS}
}

@article{zhang2021experimental,
  title={{Experimental quantum state measurement with classical shadows}},
  author={Zhang, Ting and Sun, Jinzhao and Fang, Xiao-Xu and Zhang, Xiao-Ming and Yuan, Xiao and Lu, He},
  journal={Physical Review Letters},
  volume={127},
  number={20},
  pages={200501},
  year={2021},
  publisher={APS}
}

@article{ippoliti2024classical,
  title={{Classical shadows based on locally-entangled measurements}},
  author={Ippoliti, Matteo},
  journal={Quantum},
  volume={8},
  pages={1293},
  year={2024},
  publisher={Verein zur F{\"o}rderung des Open Access Publizierens in den Quantenwissenschaften}
}

@article{li2026universal,
  title={{Universal and Efficient Quantum State Verification via Schmidt Decomposition and Mutually Unbiased Bases}},
  author={Li, Yunting and Zhu, Huangjun},
  journal={Quantum},
  volume={10},
  pages={2011},
  year={2026},
  publisher={Verein zur F{\"o}rderung des Open Access Publizierens in den Quantenwissenschaften}
}

@article{liu2019efficient,
  title = {{Efficient Verification of Dicke States}},
  author = {Liu, Ye-Chao and Yu, Xiao-Dong and Shang, Jiangwei and Zhu, Huangjun and Zhang, Xiangdong},
  journal = {Physical Review Applied},
  volume = {12},
  issue = {4},
  pages = {044020},
  year = {2019},
  publisher = {APS}
}

@article{zhu2019efficienthyp,
  title={{Efficient verification of hypergraph states}},
  author={Zhu, Huangjun and Hayashi, Masahito},
  journal={Physical Review Applied},
  volume={12},
  number={5},
  pages={054047},
  year={2019},
  publisher={APS}
}

@article{li2023robust,
  title={{Robust and efficient verification of graph states in blind measurement-based quantum computation}},
  author={Li, Zihao and Zhu, Huangjun and Hayashi, Masahito},
  journal={npj Quantum Information},
  volume={9},
  number={1},
  pages={115},
  year={2023},
  publisher={Nature Publishing Group UK London}
}

@book{nielsen2010quantum,
  title={{Quantum computation and quantum information}},
  author={Nielsen, Michael A and Chuang, Isaac L},
  year={2010},
  publisher={Cambridge University Press}
}

@misc{cha2026bmfidcode,
  author       = {Cha, Hyunho},
  title        = {Code repository},
  howpublished = {\url{https://github.com/HyunHoCha/GHZ-2loc}},
  year         = {2026}
}

@inproceedings{shafaei2013optimization,
  title={{Optimization of quantum circuits for interaction distance in linear nearest neighbor architectures}},
  author={Shafaei, Alireza and Saeedi, Mehdi and Pedram, Massoud},
  booktitle={Proceedings of the 50th Annual Design Automation Conference},
  pages={1--6},
  year={2013}
}

@article{wille2014exact,
  title={{Exact reordering of circuit lines for nearest neighbor quantum architectures}},
  author={Wille, Robert and Lye, Aaron and Drechsler, Rolf},
  journal={IEEE Transactions on Computer-Aided Design of Integrated Circuits and Systems},
  volume={33},
  number={12},
  pages={1818--1831},
  year={2014},
  publisher={IEEE}
}

@inproceedings{buadescu2019quantum,
  title={{Quantum state certification}},
  author={B{\u{a}}descu, Costin and O'Donnell, Ryan and Wright, John},
  booktitle={Proceedings of the 51st Annual ACM SIGACT Symposium on Theory of Computing},
  pages={503--514},
  year={2019}
}

@article{somma2006lower,
  title={{Lower bounds for the fidelity of entangled-state preparation}},
  author={Somma, Rolando D and Chiaverini, John and Berkeland, Dana J},
  journal={Physical Review A},
  volume={74},
  number={5},
  pages={052302},
  year={2006},
  publisher={APS}
}

@article{huang2025certifying,
  title={{Certifying almost all quantum states with few single-qubit measurements}},
  author={Huang, Hsin-Yuan and Preskill, John and Soleimanifar, Mehdi},
  journal={Nature Physics},
  volume={21},
  number={11},
  pages={1834--1841},
  year={2025},
  publisher={Nature Publishing Group UK London}
}

@article{dur2000three,
  title={{Three qubits can be entangled in two inequivalent ways}},
  author={D{\"u}r, Wolfgang and Vidal, Guifre and Cirac, J Ignacio},
  journal={Physical Review A},
  volume={62},
  number={6},
  pages={062314},
  year={2000},
  publisher={APS}
}

@article{hahn2020anonymous,
  title={{Anonymous quantum conference key agreement}},
  author={Hahn, Frederik and De Jong, Jarn and Pappa, Anna},
  journal={PRX Quantum},
  volume={1},
  number={2},
  pages={020325},
  year={2020},
  publisher={APS}
}

@article{zhang2020experimental,
  title={{Experimental optimal verification of entangled states using local measurements}},
  author={Zhang, Wen-Hao and Zhang, Chao and Chen, Zhe and Peng, Xing-Xiang and Xu, Xiao-Ye and Yin, Peng and Yu, Shang and Ye, Xiang-Jun and Han, Yong-Jian and Xu, Jin-Shi and others},
  journal={Physical Review Letters},
  volume={125},
  number={3},
  pages={030506},
  year={2020},
  publisher={APS}
}

@article{saeedi2011synthesis,
  title={{Synthesis of quantum circuits for linear nearest neighbor architectures}},
  author={Saeedi, Mehdi and Wille, Robert and Drechsler, Rolf},
  journal={Quantum Information Processing},
  volume={10},
  number={3},
  pages={355--377},
  year={2011},
  publisher={Springer}
}

@inproceedings{o2026instance,
  title={{Instance-optimal quantum state certification with entangled measurements}},
  author={O'Donnell, Ryan and Wadhwa, Chirag},
  booktitle={Proceedings of the 58th Annual ACM Symposium on Theory of Computing},
  pages={398--409},
  year={2026}
}

@article{toth2005detecting,
  title={{Detecting genuine multipartite entanglement with two local measurements}},
  author={T{\'o}th, G{\'e}za and G{\"u}hne, Otfried},
  journal={Physical Review Letters},
  volume={94},
  number={6},
  pages={060501},
  year={2005},
  publisher={APS}
}

@article{kiesel2005experimental,
  title={{Experimental analysis of a four-qubit photon cluster state}},
  author={Kiesel, Nikolai and Schmid, Christian and Weber, Ulrich and T{\'o}th, G{\'e}za and G{\"u}hne, Otfried and Ursin, Rupert and Weinfurter, Harald},
  journal={Physical Review Letters},
  volume={95},
  number={21},
  pages={210502},
  year={2005},
  publisher={APS}
}

@article{vallone2007realization,
  title={{Realization and characterization of a two-photon four-qubit linear cluster state}},
  author={Vallone, Giuseppe and Pomarico, Enrico and Mataloni, Paolo and De Martini, Francesco and Berardi, Vincenzo},
  journal={Physical Review Letters},
  volume={98},
  number={18},
  pages={180502},
  year={2007},
  publisher={APS}
}

@article{chen2007experimental,
  title={{Experimental realization of one-way quantum computing with two-photon four-qubit cluster states}},
  author={Chen, Kai and Li, Che-Ming and Zhang, Qiang and Chen, Yu-Ao and Goebel, Alexander and Chen, Shuai and Mair, Alois and Pan, Jian-Wei},
  journal={Physical Review Letters},
  volume={99},
  number={12},
  pages={120503},
  year={2007},
  publisher={APS}
}

@article{gong2019genuine,
  title={{Genuine 12-qubit entanglement on a superconducting quantum processor}},
  author={Gong, Ming and Chen, Ming-Cheng and Zheng, Yarui and Wang, Shiyu and Zha, Chen and Deng, Hui and Yan, Zhiguang and Rong, Hao and Wu, Yulin and Li, Shaowei and others},
  journal={Physical Review Letters},
  volume={122},
  number={11},
  pages={110501},
  year={2019},
  publisher={APS}
}

@article{thomas2022efficient,
  title={{Efficient generation of entangled multiphoton graph states from a single atom}},
  author={Thomas, Philip and Ruscio, Leonardo and Morin, Olivier and Rempe, Gerhard},
  journal={Nature},
  volume={608},
  number={7924},
  pages={677--681},
  year={2022},
  publisher={Nature Publishing Group UK London}
}

@article{toth2010permutationally,
  title={{Permutationally invariant quantum tomography}},
  author={T{\'o}th, G{\'e}za and Wieczorek, Witlef and Gross, David and Krischek, Roland and Schwemmer, Christian and Weinfurter, Harald},
  journal={Physical Review Letters},
  volume={105},
  number={25},
  pages={250403},
  year={2010},
  publisher={APS}
}

@article{schwemmer2014experimental,
  title={{Experimental comparison of efficient tomography schemes for a six-qubit state}},
  author={Schwemmer, Christian and T{\'o}th, G{\'e}za and Niggebaum, Alexander and Moroder, Tobias and Gross, David and G{\"u}hne, Otfried and Weinfurter, Harald},
  journal={Physical Review Letters},
  volume={113},
  number={4},
  pages={040503},
  year={2014},
  publisher={APS}
}

@article{cai2023quantum,
  title={{Quantum error mitigation}},
  author={Cai, Zhenyu and Babbush, Ryan and Benjamin, Simon C and Endo, Suguru and Huggins, William J and Li, Ying and McClean, Jarrod R and O’Brien, Thomas E},
  journal={Reviews of Modern Physics},
  volume={95},
  number={4},
  pages={045005},
  year={2023},
  publisher={APS}
}

@article{yu2025efficient,
  title={{Efficient separate quantification of state preparation errors and measurement errors on quantum computers and their mitigation}},
  author={Yu, Hongye and Wei, Tzu-Chieh},
  journal={Quantum},
  volume={9},
  pages={1724},
  year={2025},
  publisher={Verein zur F{\"o}rderung des Open Access Publizierens in den Quantenwissenschaften}
}

@article{haupt2026mitigating,
  title={{Mitigating errors in state preparation and measurement with noncomputational states}},
  author={Haupt, Conrad J and Carrera Vazquez, Almudena and Fischer, Laurin E and Woerner, Stefan and Egger, Daniel J},
  journal={Physical Review Applied},
  volume={25},
  number={2},
  pages={024079},
  year={2026},
  publisher={APS}
}

@article{khan2024multiaxis,
  title={{Multiaxis quantum noise spectroscopy robust to errors in state preparation and measurement}},
  author={Khan, Muhammad Qasim and Dong, Wenzheng and Norris, Leigh M and Viola, Lorenza},
  journal={Physical Review Applied},
  volume={22},
  number={2},
  pages={024074},
  year={2024},
  publisher={APS}
}

@article{dasgupta2024impact,
  title={{Impact of unreliable devices on stability of quantum computations}},
  author={Dasgupta, Samudra and Humble, Travis},
  journal={ACM Transactions on Quantum Computing},
  volume={5},
  number={4},
  pages={1--23},
  year={2024},
  publisher={ACM New York, NY}
}

@article{geller2020rigorous,
  title={{Rigorous measurement error correction}},
  author={Geller, Michael R},
  journal={Quantum Science \& Technology},
  volume={5},
  number={3},
  pages={03LT01},
  year={2020},
  publisher={IOP Publishing}
}

\end{document}